\documentclass[a4paper,UKenglish,cleveref, autoref, thm-restate]{lipics-v2021}

\usepackage{siunitx}
\usepackage[noend]{algpseudocode}
\graphicspath{{./FIGURES/}}   
\usepackage{amsmath} 
\usepackage{amsfonts} 
\usepackage{amssymb} 
\usepackage{xspace}   
\usepackage{relsize} 
\usepackage{verbatim} 
\usepackage{enumitem} 
\usepackage{centernot} 
\usepackage{adjustbox} 
\usepackage{multirow} 
\usepackage{stmaryrd} 
\usepackage[normalem]{ulem} 
\usepackage{ifthen} 
\usepackage{stmaryrd}
\usepackage{subcaption}
\usepackage{mathtools}
\usepackage[fleqn]{cases}

\newcommand{\stc}[1]{\text{stc}(#1)}
\newcommand{\cng}[2]{\text{cng}_{#1}(#2)}
\newcommand{\cross}[2]{\partial_{#1}(#2)}
\newcommand{\dw}[2]{#1\!:\!#2}
\newcommand{\problemSTC}{{\textsf{STC}}}
\DeclareMathSymbol{\shortminus}{\mathbin}{AMSa}{"39}
\newcommand{\problemKSTC}[1]{{#1\!\shortminus\!\textsf{STC}}}
\newcommand{\problemDSTC}[1]{{\textsf{STC}#1}}
\newcommand{\problemKDSTC}[2]{{#1\!\shortminus\!\textsf{STC}#2}}
\newcommand{\degree}[2]{\text{deg}_{#1}(#2)}
\newcommand{\neighbor}[2]{N_{#1}(#2)}
\newcommand{\eccentricity}[2]{\text{ecc}_{#1}(#2)}
\newcommand{\radius}[1]{\text{rad}(#1)}

\newcommand{\capacity}[2]{c({#1}, {#2})}
\newcommand{\flow}[2]{f({#1}, {#2})}

\newcommand{\mareksmargincomment}[1]%
    {{%
      \marginpar{{\tiny\begin{minipage}{0.5in}
                       \begin{flushleft}
                          {\color{red}MCh} {#1}
                       \end{flushleft}
                       \end{minipage}
                }}
    }}

\newcommand{\hsmargincomment}[1]%
    {{%
      \marginpar{{\tiny\begin{minipage}{0.5in}
                       \begin{flushleft}
                          {\color{blue}H} {#1}
                       \end{flushleft}
                       \end{minipage}
                }}
    }}

\newcommand{\ignore}[1]{}

\newcommand{\etal}{{\emph{et~al.}}}

\newcommand{\braced}[1]{{ \left\{ #1 \right\} }}

\newcommand{\floor}[1]{\lfloor#1\rfloor}

\newcommand{\barx}{{\bar x}}

  \newcommand{\tildec}{{\tilde{c}}}  
\newcommand{\tildex}{{\tilde{x}}}

\newcommand{\myparagraph}[1]{{\smallskip\noindent\textbf{#1}.}}
 
\newcommand{\mycase}[1]{{\underline{Case~#1}:}}

\newcommand{\PP}{{\mathbb{P}}}
\newcommand{\NP}{{\mathbb{NP}}}

\usepackage{lineno}

\title{Better Hardness Results for the Minimum Spanning Tree Congestion Problem\footnote{Research partially supported by National Science Foundation grant CCF-2153723.}}

\author{Huong Luu}{Department of Computer Science\\ University of California at Riverside}{}{}{}

\author{Marek Chrobak}{Department of Computer Science\\ University of California at Riverside}{}{}{}

\Copyright{The copyright is retained by the authors}

\authorrunning{H.~Luu and M.~Chrobak}



\begin{document}
\nolinenumbers
\maketitle


\abstract{In the spanning tree congestion problem, given a connected graph $G$, the objective is to
compute a spanning tree $T$ in $G$ that minimizes its maximum edge congestion,
where the congestion of an edge $e$ of $T$ is the number of 
edges in $G$ for which the unique path in $T$ between their endpoints traverses $e$.
The problem is known to be $\NP$-hard, but its approximability is still poorly
understood, and it is not even known whether the optimum solution can be efficiently approximated
with ratio $o(n)$. In the decision version of this problem, denoted $\problemKSTC{K}$,
we need to determine if $G$ has a spanning tree with congestion at most $K$.
It is known that $\problemKSTC{K}$ is $\NP$-complete for $K\ge 8$, and this implies 
a lower bound of $1.125$ on the approximation ratio of minimizing congestion.
On the other hand, $\problemKSTC{3}$ can be solved in polynomial time, with
the complexity status of this problem for $K\in \braced{4,5,6,7}$ remaining an open problem.
We substantially improve the earlier hardness results by proving that
$\problemKSTC{K}$ is $\NP$-complete for $K\ge 5$. This leaves only the case $K=4$
open, and improves the lower bound on the approximation ratio to $1.2$.

Motivated by evidence that minimizing congestion is hard even for graphs
of small constant radius, we consider $\problemKSTC{K}$ restricted to graphs of
radius $2$, and we prove that this variant is $\NP$-complete for all $K\ge 6$.
Exploring further in this direction, we also 
examine the variant, denoted $\problemKDSTC{K}{D}$, where the objective is to 
determine if the graph has a depth-$D$ spanning three of congestion at most $K$.
We prove that $\problemKDSTC{6}{2}$ is $\NP$-complete even for bipartite graphs.
For bipartite graphs we establish a tight bound, by also proving 
that $\problemKDSTC{5}{2}$ is polynomial-time solvable.
Additionally, we complement this result with polynomial-time algorithms for two 
special cases that involve bipartite graphs and restrictions on vertex degrees.
}


\maketitle

\section{Introduction}
\label{sec: introduction}



Problems involving constructing a spanning tree that satisfies certain requirements are among the most fundamental tasks in graph theory and algorithmics. 
One such problem is the \emph{spanning tree congestion problem}, $\problemSTC$ for short, that has been studied extensively for many years. 
In this problem we seek a spanning tree $T$ of a given graph $G$ that roughly approximates the connectivity structure of $G$, in the following
sense: Embed $G$ into $T$ by replacing each edge $(u,v)$ of $G$ by the unique $u$-to-$v$ path in $T$. Define the \emph{congestion of an edge
$e$ of $T$} as the number of such paths that traverse $e$. The objective of $\problemSTC$ is to find a spanning tree $T$ in which
the maximum edge congestion is minimized.

The general concept of edge congestion was first introduced in 1986, under the name of \emph{load factor}, as a measure of quality of an embedding of
one graph into another~\cite{sandeep_1986_optimal_tree_machines} (see also the survey in~\cite{rosenberg_1988_graph_embeddings}).
The problem of computing trees with low congestion was studied by Khuller~{\etal}~\cite{khuller_1993_designing_multi_commodity_flow_tree} 
in the context of solving commodities network routing problems. The trees
considered there were not required to be spanning subtrees, but the variant involving spanning trees was also mentioned.
In 2003, Ostrovskii provided independently a formal definition of $\problemSTC$ and established 
some fundamental properties of spanning trees with low congestion~\cite{ostrovoskii_2004_minimal_congestion_tree}.  
Since then, many combinatorial and algorithmic results about this problem have been reported in the literature
--- we refer the readers to the survey paper by Otachi~\cite{otachi_2020_survey_spanning_tree_congestion} for more
information, most of which is still up-to-date. 

As established by L{\"o}wenstein~\cite{lowenstein_2010_in_the_complement_dominating_set}, $\problemSTC$ is $\NP$-hard. 
As usual, this is proved by showing $\NP$-completeness of its decision version, where we are given a graph $G$ and an integer $K$, 
and we need to determine if $G$ has a spanning tree with congestion at most $K$. 
Otachi~{\etal}~\cite{otachi_2010_complexity_result_stc} strengthened this by proving that the problem remains $\NP$-hard even for planar graphs.
 In~\cite{okamoto_2011_hardness_results_exp_algorithm_stc}, $\problemSTC$ is proven to be $\NP$-hard for 
chain graphs and split graphs. 
On the other hand, computing optimal solutions for $\problemSTC$ can be achieved in
polynomial time for some special classes of graphs: complete $k$-partite graphs, two-dimensional tori~\cite{kozawa_2009_stc_graphs}, 
outerplanar graphs~\cite{bodlaender_2011_stc_k-outerplanargraphs}, and two-dimensional Hamming graphs~\cite{kozawa_2011_stc_rook_graphs}. 

In our paper, we focus on the decision version of $\problemSTC$ where the bound $K$ on congestion is a fixed constant. 
We denote this variant by $\problemKSTC{K}$. 
Several results on the complexity of $\problemKSTC{K}$ were reported in~\cite{otachi_2010_complexity_result_stc}.
For example, the authors of~\cite{otachi_2010_complexity_result_stc}
show that $\problemKSTC{K}$ is decidable in linear time for planar graphs, graphs of bounded treewidth, graphs of bounded degree, and for all graphs when $K=1,2,3$. On the other hand, they show that the problem is $\NP$-complete for any fixed $K \ge 10$. 
In~\cite{bodlaender_2012_parameterized_complexity_stc}, Bodlaender~\etal~proved that $\problemKSTC{K}$ is linear-time solvable for graphs in apex-minor-free families and chordal graphs. 
They also show an improved hardness result of $\problemKSTC{K}$, namely that it is $\NP$-complete for $K \ge 8$, even in the special case of apex graphs that only have one unbounded degree vertex. 
As stated in~\cite{otachi_2020_survey_spanning_tree_congestion}, the complexity status of $\problemKSTC{K}$ for $K \in\braced{4,5,6,7}$ remains an open problem.

Very little is known about the approximability of $\problemSTC$.
The trivial upper bound for the approximation ratio is $n/2$ ---
this ratio is achieved in fact by \emph{any} spanning tree~\cite{otachi_2020_survey_spanning_tree_congestion}.
As a direct consequence of the $\NP$-completeness of  $\problemKSTC{8}$,
there is no polynomial-time algorithm to approximate the optimum spanning tree congestion with a ratio better than $1.125$  (unless $\PP = \NP$).  


\myparagraph{Our contributions} 
In this paper, 
addressing an open question in~\cite{otachi_2020_survey_spanning_tree_congestion}, we provide an improved hardness result for $\problemKSTC{K}$:

\smallskip

\begin{theorem}\label{thm: np-hardness}
For any fixed integer $K \ge 5$,  $\problemKSTC{K}$ is  $\NP$-complete. 
\end{theorem}

\smallskip

The proof of this theorem is given in Section~\ref{sec: k-stc hardness proof}.
Combined with the results in~\cite{otachi_2010_complexity_result_stc}, Theorem~\ref{thm: np-hardness}
leaves only the status of $\problemKSTC{4}$ open. 
Furthermore, it also immediately improves the lower bound on the approximation ratio for $\problemSTC$:

\smallskip

\begin{corollary}
For $c < 1.2$ there is no polynomial-time $c$-approximation algorithm for $\problemSTC$, unless $\PP=\NP$.
\end{corollary}

\smallskip

We remark that this hardness result remains valid even if an additive constant is allowed in the
approximation bound. This follows by an argument in~\cite{bodlaender_2012_parameterized_complexity_stc}.
(In essence, the reason is that assigning a positive integer weight $\beta$ to each edge increases its congestion 
by a factor $\beta$.)

\smallskip
A common feature of the hardness proofs for $\problemSTC$, including ours, is that they all
use graphs of small constant radius (or, equivalently, diameter). Another property of $\problemSTC$
that makes its approximation challenging is that the minimum congestion value is not monotone with respect to
adding edges. The example graph in~\cite{ostrovoskii_2004_minimal_congestion_tree} showing this non-monotonicity
is also of small radius (in fact, only $2$). These observations indicate that
a key to further progress may be in better understanding of  $\problemSTC$ in small-radius graphs.

This motivates our additional hardness result presented in Section~\ref{sec: k-stc hardness bipartite graph radius 2}, where we focus
on graphs of radius $2$. (For radius $1$ the problem is trivial.)
We prove there that $\problemKSTC{K}$ remains $\NP$-complete for this class of graphs, for any fixed integer $K \ge 6$.
In fact, this holds even if we further restrict such graphs to be bipartite and have only one vertex of non-constant degree.

Probing further in this direction, in Section~\ref{sec: k-dstc-2} we consider the variant of 
$\problemSTC$ denoted $\problemKDSTC{K}{D}$, in which the objective is to determine if the graph has a 
spanning tree of depth at most $D$ and congestion at most $K$. Note that this is not a restriction of $\problemSTC$,
as the minimum congestion for trees of depth $2$ can be larger than then optimum value of $\problemSTC$.
We observe that our $\NP$-completeness proof in Section~\ref{sec: k-stc hardness bipartite graph radius 2}
can be adapted to prove that $\problemKDSTC{K}{2}$ is $\NP$-complete for $K\ge 6$. 
This is true even if input graphs are restricted to bipartite graphs with only one vertex of non-constant degree.
For bipartite graphs, we establish a tight bound by proving that $\problemKDSTC{5}{2}$ is polynomial-time solvable.
Complementing this, we present two other natural special cases, involving bipartite graphs and restrictions on vertex degrees,
in which the optimal congestion spanning tree can be computed in polynomial time.


\myparagraph{Other related work} 
The spanning tree congestion problem is closely related to the tree spanner problem, in which the objective is to find
 a spanning tree $T$ of $G$ that minimizes the stretch factor, defined as the maximum ratio, over all vertex pairs,
 between the length of the path in $T$ and the length of the shortest path in $G$ connecting these vertices.   
In fact, for any planar graph, its spanning tree congestion is equal to its dual's minimum stretch factor 
plus one~\cite{fakete_2001_tree_spanner,otachi_2010_complexity_result_stc}. 
This direction of research has been extensively explored, see~\cite{cai_1995_tree_spanner,dragan_2011_spanner_in_sparse_graph,yuval_2009_approx_min_max_stretch}. 
As an aside, we remark that the complexity of the tree $3$-spanner problem has been open since its first introduction in 1995~\cite{cai_1995_tree_spanner}. 

$\problemSTC$ is also intimately related to problems involving cycle bases in graphs.
As each spanning tree induces a fundamental cycle basis of the given graph, a spanning tree with low congestion yields a cycle basis 
for which the edge-cycle incidence matrix is sparse. Sparsity of such matrices is desirable in linear-algebraic approaches to solving
some graph optimization problems, for example analyses of
distribution networks such as pipe flow systems~\cite{alvarruiz_2015_improving_efficiency_loop_method}. 

$\problemSTC$ can be thought of as an extreme case of the graph sparsification problem, where, given a graph $G$, the objective is
 to compute a sparse graph $H$ 
that captures connectivity properties of $G$. Such $H$ can be used instead of $G$ for the purpose of various analyses, to improve efficiency.
See~\cite{benczur_1996_approximating_st_min_cut,fung_2011_framework_graph_sparsification,spielman_2011_spectral_sparsification} 
(and the references therein) for some approaches to graph sparsification.


\section{Preliminaries}
\label{sec: preliminaries}


\myparagraph{Basic graph terminology}
Let $G$ be a simple graph with vertex set $V$ and edge set $E$. 
We use notation $\neighbor{G}{v}$ for the neighborhood of a vertex $v \in V$ and $\degree{G}{v}$ for its degree. 
For a vertex $v \in V$, its \emph{eccentricity} $\eccentricity{G}{v}$ is defined as the maximum distance from $v$ to any other vertex. 
The \emph{radius} of $G$ is $\radius{G} = \min_{v \in V} \eccentricity{G}{v}$. 

Consider a spanning tree $T\subseteq E$ of $G$. If $e = (u,v)\in T$, removing $e$ from $T$ splits $T$ into two subtrees. 
We denote by $T_{u,v}$ the subtree that contains $u$ and by $T_{v,u}$ the subtree that contains $v$. 
Let the \emph{cut-set} of $e$, denoted $\cross{G,T}{e}$, be the set of edges in $E$ that have one endpoint in $T_{u,v}$ and the other in $T_{v,u}$. 
In other words, $\cross{G,T}{e}$ consists of the edges $(u',v') \in E$ for which the unique (simple) path in $T$ from $u'$ to $v'$
goes through $e$. Note that $e \in \cross{G,T}{e}$.
The \emph{congestion of $e$}, denoted by $\cng{G,T}{e}$, is the cardinality of $\cross{G,T}{e}$. 
The \emph{congestion of tree $T$} is $\cng{G}{T} = \max_{e \in T} \cng{G,T}{e}$. 
Finally, the \emph{spanning tree congestion of graph $G$}, denoted by $\stc{G}$, is defined as the minimum value of 
$\cng{G}{T}$ over all spanning trees $T$ of $G$. 


\myparagraph{Weighted edges}
The concept of the spanning tree congestion extends naturally to edge-weighted graphs. An edge $e$ with integer weight $\omega\ge 1$
contributes $\omega$ to the congestion of any edge $f$ for which $e \in \cross{G,T}{f}$. One can think of $e$ as representing
$\omega$ parallel edges between $u$ and $v$. We refer to these parallel edges as a \emph{non-weighted realization} of a weighed edge $e$.
Indeed, replacing $e$ by this realization does not affect the minimum congestion value, because in a multigraph
only one edge between any two given 
vertices can be in a spanning tree, but all of them belong to the cut-set $\cross{G,T}{f}$ of
any edge $f\in T$ whose removal separates these vertices in $T$ (and thus all contribute to $\cng{G,T}{f}$). 

We can also realize a weighted edge using a simple graph (without multiple edges). 
As observed in~\cite{otachi_2010_complexity_result_stc} (and is easy to prove), 
edge subdivision does not affect the spanning tree congestion of a graph, so,
instead of using parallel edges we can realize an edge of weight $\omega$ using $\omega$ parallel disjoint paths. 
(See Figure~\ref{fig: single weight} for illustration.) 
We state our results in terms of simple graphs, but we use weighted graphs in our proofs with the understanding
that they actually represent simple graphs.
As all weights used in the paper are constant, the computational complexity of $\problemKSTC{K}$ is not affected.
The proof in Section~\ref{sec: k-stc hardness proof} does not depend on what realization of weighted edges we use,
while the proof in Section~\ref{sec: k-stc hardness bipartite graph radius 2} uses a specific realization that
we refer to as \emph{spintop}:  an edge $(u,v)$ of weight $\omega$ is realized using
$\omega-1$ length-three $u$-to-$v$ paths in addition to a non-weighted edge $(u,v)$ itself (see Figure~\ref{fig: single weight}b).

\begin{figure}[ht]
	\begin{center}
		\includegraphics[height = 1.5in]{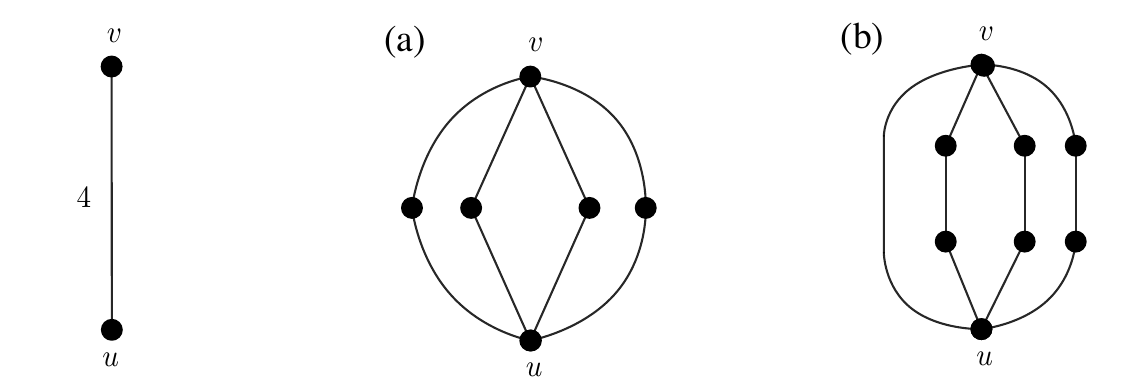}
	\end{center}
	\caption{Two different realizations of an edge $(u,v)$ of multiplicity $4$. 
	(a) A basic realization using paths of length $2$.
	(b) The spintop realization used in Section~\ref{sec: k-stc hardness bipartite graph radius 2}.}
	\label{fig: single weight}
\end{figure}

\smallskip


\myparagraph{Double weights}
In fact, it is convenient to generalize this further by introducing edges with \emph{double weights}. A double
weight of an edge $e$ is denoted $\dw{\omega}{\omega'}$, where $\omega$ and $\omega'$ are positive integers such that $\omega \le \omega' \le K - 1$, 
and its interpretation in the context of $\problemKSTC{K}$ is as follows: Given a spanning tree $T$,
\begin{itemize}[label=$\circ$]
	\item if $e\in E\setminus T$ then $e$ contributes $\omega$ to the congestion $\cng{G,T}{f}$
			of any edge $f\in T$ for which $e\in \cross{G,T}{f}$, and 
	\item if $e\in T$ then $e$ contributes $\omega'$ to its own congestion, $\cng{G,T}{e}$. 
\end{itemize}
The lemma below provides a simple-graph realization of double-weighted edges. It implies that
including such edges does not affect the computational
complexity of $\problemKSTC{K}$, allowing us to formulate our proofs in terms of graphs where some edges have double weights.

%
\medskip
\begin{lemma}\label{lemma: double weights}
Let $(u,v)$  be an edge in $G$ with double weight $\dw{\omega}{\omega'}$, where $\omega \le \omega'  \le K - 1$.  
Consider another graph $G'$ obtained from $G$ by 
removing $(u,v)$, and for each $i = 1,2,\cdots, \omega$ adding a new vertex $w_i$ with two edges:
edge  $(u, w_i)$ of weight $1$ and edge $(w_i, v)$ of weight $\omega'-\omega+1$ (see Figure~\ref{fig: double weights}a for an example). 
Then, $\stc{G} \le K$ if and only if $\stc{G'} \le K$. 
\end{lemma}

\begin{figure}[ht]
	\begin{center}
		\includegraphics[width = 5in]{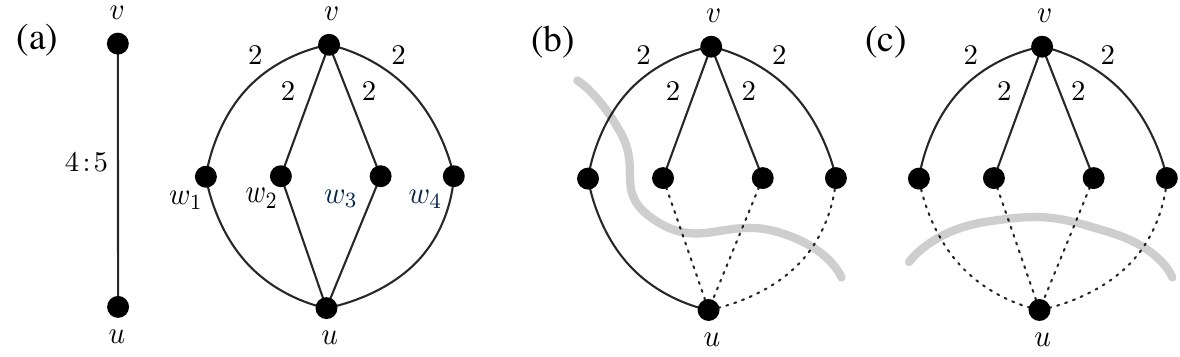}
	\end{center}
	\caption{(a) On the left, an edge $(u,v)$ with double weight $\dw{4}{5}$ in $G$. On the right, the realization of $(u,v)$ in $G'$. 
	If one applies the spintop realization of the edges from $v$ to $w_i$'s, as in Figure~\ref{fig: single weight}b, then
	the subgraph on the right realizing $(u,v)$ is bipartite and all its nodes are within distance $2$ from $v$.
	Figures~(b) and~(c) illustrate the proof of Lemma~\ref{lemma: double weights}:
	 (b) the traversal of $T'$ and the cut of $(u,v)$ when $(u,v)\in T$,  
	 (c) the traversal of $T'$ and the cut containing $(u,v)$ when $(u,v) \notin T$. 
	 Solid lines are tree edges and dotted lines are non-tree edges.}
	\label{fig: double weights}
\end{figure}

\begin{proof}
Denote by $W = \braced{w_1, w_2 ,..., w_\omega}$ the set of new vertices, and by $W_u = \braced{(u,w_i) \; |\; w_i \in W}$ 
and $W_v= \braced{(w_i, v)\; |\; w_i \in W}$ the sets of new edges added to $G'$.

\smallskip \noindent
($\Rightarrow$)
Suppose that $G$ has a spanning tree $T$ with $\cng{G}{T} \le K$. 
We will show that there exists a spanning tree $T'$ of $G'$ with $\cng{G'}{T'} \le K$.
We break the proof into two cases, in both cases showing that $\cng{G',T'}{e} \le K$ for each edge $e\in T'$.

\smallskip\noindent
\mycase{1} $(u,v) \in T$.  

Consider the spanning tree $T' = T  \setminus \braced{(u,v)} \cup W_v \cup \braced{(w_1, u)} $ of $G'$ (see Figure~\ref{fig: double weights}b). 
For every edge $(x,y) \in E \setminus \braced{(u,v)} $, the $x$-to-$y$ paths in $T$ and $T'$ are the same, except that if the
$x$-to-$y$ path in $T$ traverses edge $(u,v)$ then the $x$-to-$y$ path in $T'$ traverses $(u,w_1), (w_1,v)$ instead. Therefore, 
\begin{itemize}[label=$\circ$]
	\item If $e \in T' \setminus (W_v \cup \braced{(u, w_1)})$, then $\cross{G',T'}{e} = \cross{G,T}{e}$.
		 So $\cng{G',T'}{e}  = \cng{G,T}{e} \le K$. 
	\item If $e = (u, w_1)$, then $\cross{G',T'}{e} = \cross{G,T}{u,v} \setminus \braced{(u,v)} \cup W_u$. 
		By the definition of double weights, $(u,v)$ contributes $\omega'$ to $\cng{G,T}{u,v}$ while each edge in $W_u$ contributes $1$ to $\cng{G',T'}{e}$. 
		Hence, $\cng{G',T'}{e} = \cng{G,T}{u, v} - \omega' + \omega \le \cng{G,T}{e} \le K$.  
	\item If $e = (w_1, v)$, then $\cross{G',T'}{e} = \cross{G,T}{u,v}  \setminus \braced{(u,v)} \cup \braced{e} \cup (W_u \setminus \braced{(w_1, u)})$. 
		Since $e$ contributes $\omega' - \omega  +1$ to its own congestion, we have: $\cng{G',T'}{e} = \cng{G,T}{u, v} - \omega' + (\omega' - \omega + 1) + (\omega  -1) = \cng{G,T}{u, v} \le K$.  
	\item Lastly, if $e \in W_v \setminus \braced{(w_1, v)}$ then it is a leaf edge, 
			we have $\cng{G',T'}{e} = \omega' - \omega + 2 \le \omega'+ 1\le K$. 
\end{itemize}

\smallskip\noindent
\mycase{2} $(u,v)\notin T$. 

Let $T' = T  \cup W_v$, which is a spanning tree of $G'$ (see Figure~\ref{fig: double weights}c). We consider the following sub-cases: 
\begin{itemize}[label=$\circ$]
	\item If $e \in W_v$ then, as $e$ is a leaf edge, we have $\cng{G',T'}{e} = \omega' - \omega + 2 \le  \omega'+ 1\le K$.
	\item If $e \in T' \setminus W_v$ and $e$ is not on the $u$-to-$v$ path in $T'$, 
						then $\cross{G',T'}{e} = \cross{G,T}{e}$. So $\cng{G',T'}{e} = \cng{G,T}{e} \le K$.
	\item If $e \in T' \setminus W_v$ and $e$ is on the $u$-to-$v$ path in $T'$,
		 				then $\cross{G',T'}{e} = \cross{G,T}{e} \setminus \braced{(u,v)} \cup W_u$. 
							Since $(u,v)$ contributes $\omega$ to $\cng{G,T}{e}$ and $W_u$ also contributes $\omega$ to $\cng{G',T'}{e}$, 
			we have $\cross{G',T'}{e} = \cross{G,T}{e} \le K$.
\end{itemize}

\noindent
We have shown that $\cng{G'}{T'} \le K$ in all cases, which completes the proof for the forward implication. 
We now proceed to the proof of the converse implication. 

\smallskip 
\noindent
($\Leftarrow$) 
Let $T'$ be the spanning tree of $G'$ with congestion $\cng{G'}{T'} \le K$.
We will show that there exists a spanning tree $T$ of $G$ with $\cng{G}{T} \le K$. 
Note that, for any $w_i \in W$, $T'$ traverses at least one of the two edges $(u,w_i)$ and $(w_i, v)$. 
Furthermore, at most one vertex in $W$ is a non-leaf.  We consider three cases.
In the first two cases the arguments follow the same pattern as in the proof for the $(\Rightarrow)$ implication,
in essence reversing the modification of the spanning tree. Then the third case reduces to the second case.

\smallskip\noindent
\mycase{1} Exactly one vertex in $W$ is a non-leaf in $T'$. 

\smallskip
Without loss of generality, we can assume $w_1$ is a non-leaf vertex (that is, both $(u,w_1)$ and $(w_1,v)$ are in $T$)
and $W \setminus \braced{w_1}$ are leaves. 
We construct $T$ by adding $(u,v)$ to $T'$ and removing all vertices of $W$ and their incident edges from $T$. 
By the construction, $T$ is a spanning tree of $G$. We have: 
\begin{itemize}[label=$\circ$]
	\item If $e \in T \setminus \braced{(u,v)}$, then $\cng{G,T}{e} = \cng{G',T'}{e}\le K$. 
	\item If $e=(u,v)$, then $\cng{G,T}{e} \le \cng{G',T'}{v, w_1}\le K$. 
\end{itemize}

\smallskip\noindent
\mycase{2} All vertices in $W$ are leaves and $T'$ traverse all edges in $W_v$.

\smallskip 
Let $T = T' \setminus W_v$, which is a spanning tree of $G$. Then
\begin{itemize}[label=$\circ$]
	\item If $e \in T$ and $e$ is not on the $u$-to-$v$ path in $T$, then $\cng{G,T}{e} = \cng{G',T'}{e}\le K$.
	\item If $e \in T$ and $e$ is on the $u$-to-$v$ path in $T$, then
			$(u,v)$ and $W_u$ contribute the same amount $\omega$ to the congestion of $e$ in $T$ and $T'$, respectively,
			implying that $\cng{G,T}{e} = \cng{G',T'}{e} \le K$. 
\end{itemize}

\smallskip\noindent
\mycase{3} All vertices in $W$ are leaves and $T'$ traverses at least one edge in $W_u$.

\smallskip 
In this case, we consider another spanning tree $T''$ of $G'$ that traverses all edges in $W_v$ and does not use any edge in $W_u$. 
It is sufficient to show that  $\cng{G'}{T''} \le \cng{G'}{T'}$, since it implies that
$\cng{G'}{T''} \le K$, and then we can apply Case~2 to $T''$.
We examine the congestion values of each edge $e \in T''$:
\begin{itemize}[label=$\circ$]
	\item If $e \in T'' \setminus W_v$ and $e$ is not on the $u$-to-$v$ path in $T''$, 
				then $e \in T'$ and $\cross{G',T''}{e} =  \cross{G',T'}{e}$, implying $\cng{G',T''}{e} =  \cng{G',T'}{e}$. 
	\item If $e \in T'' \setminus W_v$ and $e$ is on the $u$-to-$v$ path in $T''$, 
			then for each vertex $w_i \in W$ either $(u, w_i)$ contributes $1$ or $(w_i, v)$ contributes $\omega'-\omega+1 \ge 1$ to $\cng{G',T'}{e}$. 
		On the other hand, in $T''$, all edges in $W_u$ are in $\cross{G',T''}{e}$ and contribute a total of $\omega$ to $\cng{G',T''}{e}$. 
		Thus, $\cng{G', T''}{e} \le \cng{G', T'}{e}$. 
	\item If $e \in W_v$, then  $\cng{G',T''}{e} = \omega' - \omega +2 \le \omega'+1 \le K$.
\end{itemize}

\noindent 
In all cases, we have proved that there is a spanning tree $T$ of $G$ that has congestion 
at most $K$ establishing the validity of the backward implication. 
\end{proof}


As explained earlier, in Section~\ref{sec: k-stc hardness bipartite graph radius 2} we will use the spintop realization
for weighted edges. With this, the realization of an edge $e = (u,v)$ with double weight $\dw{\omega}{\omega'}$ will use
the spintop realization for the edges of weight $\omega'-\omega+1$ between $v$ and the $w_i$'s.
The property of this realization that will be crucial in Section~\ref{sec: k-stc hardness bipartite graph radius 2}
is that it is bipartite and all its nodes are within distance $2$ from $v$.

\smallskip
\emph{Remark:}
Some readers may have noticed that there is a simpler way to realize an edge $(u,v)$
with a double weight $\dw{\omega}{\omega'}$: replace it by a length-$2$ path
$(u,w)$, $(w,v)$, where $w$ is a new vertex, edge $(u,w)$ has weight $\omega$, and edge $(w,v)$ has weight $\omega'$.
This indeed works, but can be used only when $\omega+\omega' \le K$. This is because, in this construction,
if $w$ is a leaf of a spanning
tree, the congestion of the tree edge from $w$ will be $\omega+\omega'$, and this congestion value cannot exceed $K$.
This realization of double-weighted edges would suffice for our proof in Section~\ref{sec: k-stc hardness proof},
but not the one in Section~\ref{sec: k-stc hardness bipartite graph radius 2}. (It may also be useful for
establishing other hardness results for $\problemSTC$.)


\section{$\NP$-completeness proof of $\problemKSTC{K}$ for $K \ge 5$}
\label{sec: k-stc hardness proof}

In this section we prove our main result, the $\NP$-completeness of  $\problemKSTC{K}$. Our proof uses an $\NP$-complete
variant of the satisfiability problem called (2P1N)-SAT~\cite{dahlhaus_1994_complexity_multiterminal_cuts,yoshinaka_2005_2p1n_sat}. An instance of (2P1N)-SAT
is a boolean expression $\phi$ in conjunctive normal form, where each variable occurs exactly three times,
twice positively and once negatively, and each clause contains exactly two or three literals of different variables.
The objective is to decide if $\phi$ is satisfiable, that is if there is a satisfying assignment that makes $\phi$ true.

For each constant $K$, $\problemKSTC{K}$ is clearly in $\NP$. We will present a polynomial-time reduction from (2P1N)-SAT.
In this reduction, given an instance $\phi$ of (2P1N)-SAT, we construct a graph $G$ with the following property:

\smallskip
\begin{description}
	\item{$(\ast)$} $\phi$ has a satisfying truth assignment if and only if $\stc{G} \le K$. 
\end{description}

\smallskip

Throughout the proof, the three literals of $x_i$ in $\phi$ will be denoted by $x_i$, $x'_i$, and $\barx_i$, where
$x_i$, $x'_i$ are the two positive occurrences of $x_i$ and $\barx_i$ is the negative occurrence of $x_i$.
We will also use notation $\tildex_i$ to refer to an unspecified literal of $x_i$, that is $\tildex_i\in\braced{x_i,x'_i,\barx_i}$.

We now describe the reduction. 
Set $k_i = K-i$ for $i = 1,2,3,4$. (In particular, for $K=5$, we have $k_1 = 4$, $k_2 = 3$, $k_3 = 2$, $k_4 = 1$). 
$G$ will consist of gadgets corresponding to variables, with the gadget corresponding to $x_i$ having three vertices $x_i$, $x'_i$, and $\barx_i$,
that represent its three occurrences in the clauses. 
$G$ will also have vertices representing clauses and edges connecting
literals with the clauses where they occur (see Figure~\ref{fig: graph example}b for an example). As explained in Section~\ref{sec: preliminaries}, without any loss of generality we can allow edges in $G$ to have constant-valued weights, single or double.  Specifically, 
starting with $G$ empty, the construction of $G$ proceeds as follows:
\begin{itemize}[label=$\circ$]
	\item Add a \emph{root vertex} $r$. 
	\item For each variable $x_i$, construct the \emph{$x_i$-gadget} (see Figure~\ref{fig: graph example}a).
		 This gadget has three vertices corresponding to the literals: 
		 a \emph{negative literal vertex} $\barx_{i}$ and two \emph{positive literal vertices} $x_i, x_i'$,
		 and two auxiliary vertices $y_i$ and $z_i$. Its edges and their weights are given in the table below:
		\smallskip
		 \begin{center}
		 	\setlength\tabcolsep{2pt}
		 	\begin{tabular*}{\linewidth}{@{\extracolsep{\fill}} |c||c|c|c|c|c|c|c| } \hline
				edge &  $(\barx_i, z_i)$ \ \ &  $(z_i, x_i)$ \ \  &  $(x_i, x_i') $ \  \ &  $(r,x_i')$ \ \ &  $(r, y_i)$ \  \ &  $(y_i, z_i)$ \  \ &  $(y_i, \bar{x}_i)$ \ \ 
				\\ \hline
				weight & $\dw{1}{k_3}$ \ & $\dw{1}{k_3}$ \ & $\dw{1}{k_2}$ \ & $k_3$ \ & $k_4$ \ & $k_4$ \ & $\dw{1}{k_2}$ \
				\\ \hline
			\end{tabular*}
		\end{center}
		\smallskip
	\item For each clause $c$, create a \emph{clause vertex} $c$.  For each literal $\tildex_i$ in $c$, add 
	the corresponding \emph{clause-to-literal edge} $(c,\tildex_i)$ of weight $\dw{1}{k_2}$.
	Importantly, as all literals in $c$ correspond to different variables, these edges will go to different variable gadgets.
	\item For each two-literal clause $c$, add a \emph{root-to-clause} edge $(r,c)$ of weight $\dw{1}{k_1}$.
\end{itemize}

\begin{figure}[ht]
	\begin{center}
		\includegraphics[height = 1.8in]{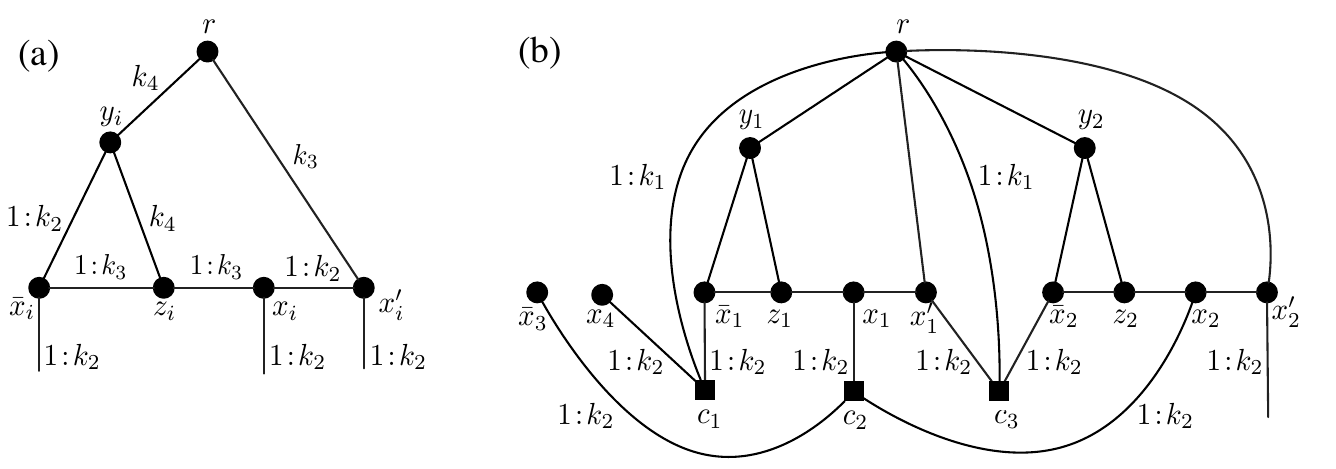}
	\end{center}
	\caption{(a)The $x_i$-gadget. (b) An example of a partial graph $G$ for the boolean expression 
	$\phi =  c_1 \land c_2 \land c_3 \land \cdots$ where 
	$c_1 = \barx_1 \lor x_4$, $c_2 =  x_1 \lor x_2 \lor \barx_3$, and $c_3 = x_1 \lor \barx_2$.
	(The weights of edges inside the variable gadgets are not shown.)
	}
	\label{fig: graph example}
\end{figure}

We now show that $G$ has the required property $(\ast)$, proving the two implications separately.


$(\Rightarrow)$ Suppose that $\phi$ has a satisfying assignment. Using this assignment, we
construct a spanning tree $T$ of $G$ as follows: 
\begin{itemize}[label=$\circ$]
	\item For every $x_i$-gadget, include in $T$ edges $(r, x_i')$, $(r, y_i)$, and $(y_i, z_i)$. 
		If $x_i = 0$,  include in $T$ edges $(\bar{x}_i, z_i)$ and $(x_i, x_i')$, otherwise include in $T$ edges $(y_i, \bar{x}_i)$ and $(z_i, x_i)$. 
	\item For each clause $c$, include in $T$ one clause-to-literal edge that is incident to any literal vertex that satisfies $c$
		in our chosen truth assignment for $\phi$.
\end{itemize}

By routine inspection, $T$ is indeed a spanning tree of $G$: Each $x_i$-gadget is traversed from $r$ without cycles, and all
clause vertices are leaves of $T$. 
Figures~\ref{fig: neg spanning tree} and~\ref{fig: pos spanning tree} show how $T$ traverses an $x_i$-gadget in different cases, depending
on whether $x_i = 0$ or $x_i = 1$ in the truth assignment for $\phi$, and on which literals are chosen to satisfy each clause. 
Note that the edges with double weights satisfy the assumption of Lemma~\ref{lemma: double weights} in Section~\ref{sec: preliminaries},
that is each such weight $\dw{1}{\omega'}$ satisfies $1 \le \omega' \le K - 1$. 

We need to verify that each edge in $T$ has congestion at most $K$. 
All the clause vertices are leaves in $T$, thus the congestion of each clause-to-literal edge is $k_2 + 2 = K$ (this
holds for both three-literal and two-literal clauses). 
To analyze the congestion of the edges inside an $x_i$-gadget, we consider two cases, depending on the value of $x_i$ in our truth assignment.

When $x_i = 0$, we have two sub-cases (a) and (b) as shown in Figure~\ref{fig: neg spanning tree}. The congestions of the edges in the $x_i$-gadget are as follows:
\begin{itemize}[label=$\circ$]
	\item In both cases, $\cng{G, T}{r, x_i'}= k_3 + 3$.  
	\item In case~(a), $\cng{G, T}{r,y_i} = k_4 + 3$.  In case~(b), it is $k_4 + 2$.    
	\item In case~(a), $\cng{G, T}{y_i,z_i} = k_4 + 4$. In case~(b), it is $k_4 + 3$.  
	\item In case~(a), $\cng{G, T}{\barx_i, z_i} = k_3 +3$. In case~(b), it is  $k_3 + 2$.  
	\item In both cases, $\cng{G, T}{x_i, x_i'} = k_2+2$.
\end{itemize}
%

\begin{figure}[ht]
	\begin{center}
		\includegraphics[height = 1.8in]{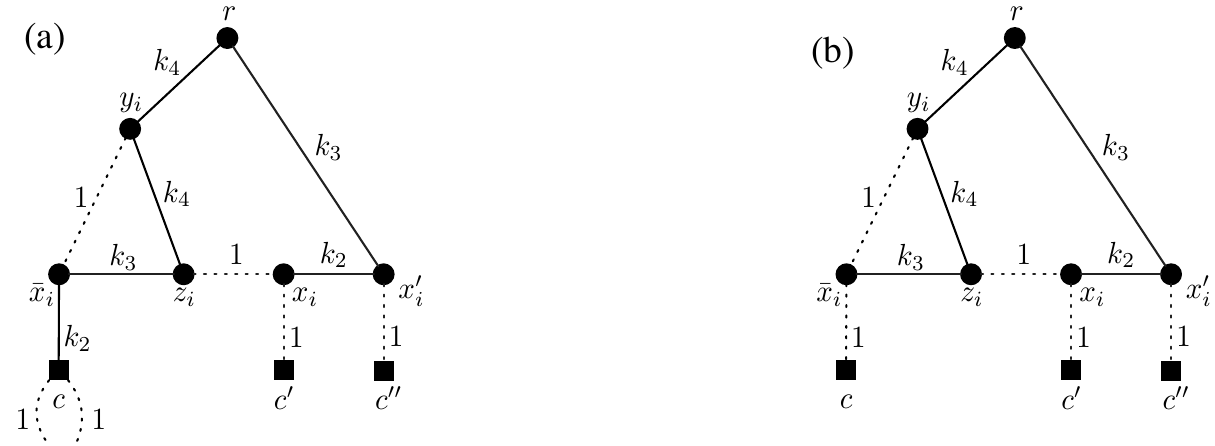}
	\end{center}
	\caption{The traversal of the $x_i$-gadget by $T$ when $x_i = 0$. Solid lines are tree edges, dotted lines are non-tree edges. 
	(a) $\barx_i$ is chosen by clause $c$. 
	(b) $\barx_i$ is not chosen by clause $c$. }
	\label{fig: neg spanning tree}
\end{figure}

On the other hand, when $x_i = 1$, we have four sub-cases. Figure~\ref{fig: neg spanning tree} illustrates cases (a)--(c). 
In case (d) (not shown in Figure~\ref{fig: neg spanning tree}), none of the positive literal vertices $x_i, x_i'$ is chosen to satisfy their corresponding clauses. 
The congestions of the edges in the $x_i$-gadget are as follows:
\begin{itemize}[label=$\circ$]
	\item In cases~(a) and (b), $\cng{G, T}{r, x_i'}=k_3 +3$. In cases~(c) and~(d), it is $k_3 + 2$. 
	\item In cases~(a) and (c), $\cng{G, T}{r,y_i}  = k_4 + 4$. In cases~(b) and~(d), it is $k_4 + 3$.   
	\item In cases~(a) and (c), $\cng{G, T}{y_i, z_i}  = k_4 + 4$. In cases~(b) and~(d), it is $k_4 + 3$.   
	\item In cases~(a) and (c), $\cng{G, T}{z_i,x_i}  = k_3 + 3$. In cases~(b) and~(d), it is $k_3 + 2$.   
	\item In all cases, $\cng{G, T}{y_i, \barx_i} = k_2 +2$. 
\end{itemize}
%

\begin{figure}[ht]
	\begin{center}
		\includegraphics[height = 1.8in]{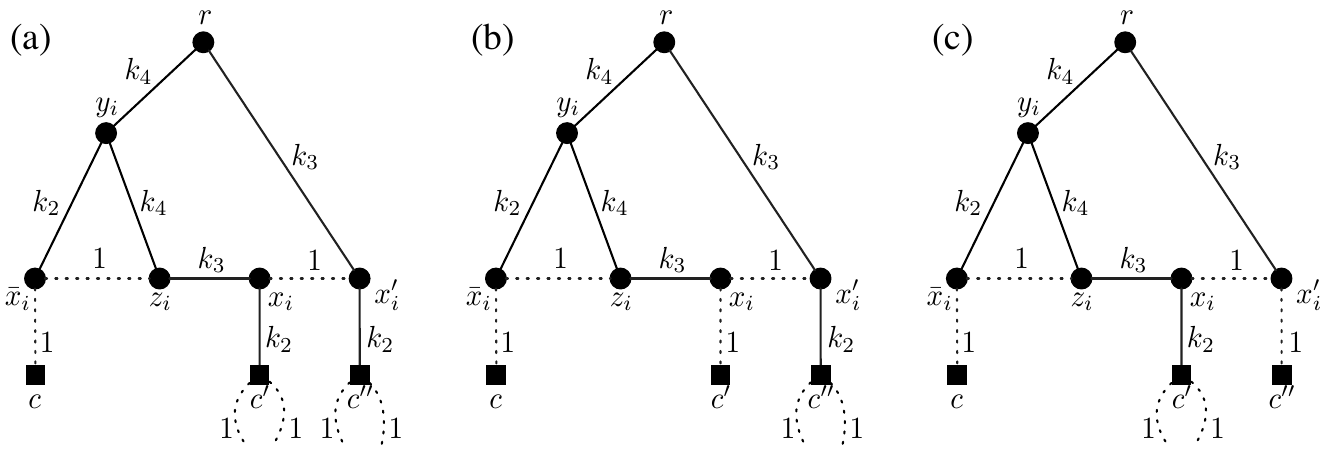}
	\end{center}
	\caption{The traversal of the $x_i$-gadget by $T$ when $x_i = 1$. 
	By $c$, $c'$ and $c''$ we denote the clauses that contain literals $\barx_i$, $x_i$ and $x'_i$, respectively. 
	(a) $x_i$ and $x'_i$ are chosen by clauses $c'$ and $c''$. 
	(b) $x'_i$ is chosen by clause $c''$. 
	(c) $x_i$ is chosen by clause $c'$. }
	\label{fig: pos spanning tree}
\end{figure}

In summary, the congestion of each edge of $T$ is at most $K$. Thus $\cng{G}{T} \le K$; in turn, $\stc{G} \le K$, as claimed.


$(\Leftarrow)$
We now prove the other implication in~$(\ast)$. We assume that $G$ has a spanning tree $T$ with $\cng{G}{T} \le K$. 
We will show how to convert $T$ into a satisfying truth assignment for $\phi$.
The proof consists of a sequence of claims showing that $T$ must have a special form that will allow us to
define this truth assignment.
 

\medskip

\begin{claim}
\label{claim: root to literal path}
Each $x_i$-gadget satisfies the following property: for each literal vertex $\tildex_i$,
if some edge $e$ of $T$ (not necessarily in the $x_i$-gadget) is on the $r$-to-$\tildex_i$ path in $T$, 
then $\cross{G,T}{e}$ contains at least two distinct edges from this gadget other than $(y_i, z_i)$. 
\end{claim}

\medskip

This claim is straightforward: it follows directly from the fact that there are two edge-disjoint paths from $r$ 
to any literal vertex $\tildex_i\in\braced{\barx_i, x_i, x_i'}$ that do not use edge $(y_i, z_i)$.

%

\medskip

\begin{claim}\label{claim: root-to-clause edge}
	For each two-literal clause $c$,  edge $(r,c)$ is not in $T$. 
\end{claim}
 
\medskip

For each literal $\tildex_i$ of clause $c$, there is an $r$-to-$c$ path via the $x_i$-gadget, so,
together with edge $(r,c)$, $G$ has three disjoint $r$-to-$c$ paths.
Thus, if $(r,c)$ were in $T$, its congestion would be at least $k_1+2>K$, proving Claim~\ref{claim: root-to-clause edge}.

%

\medskip

\begin{claim}\label{claim: clause vertex}
	All clause vertices are leaves in $T$. 
\end{claim}

\medskip

To prove Claim~\ref{claim: clause vertex},
suppose there is a clause $c$ that is not a leaf. Then, by Claim~\ref{claim: root-to-clause edge},
$c$ has at least two clause-to-literal edges in $T$, say $(c, \tildex_i)$ and $(c, \tildex_j)$. 
We can assume that the last edge on the $r$-to-$c$ path in $T$ is $e = (c, \tildex_i)$. 
Clearly, $r \in T_{\tildex_i,c}$ and $\tildex_j \in T_{c,\tildex_i}$. 
By Claim~\ref{claim: root to literal path}, at least two edges of the $x_j$-gadget are in $\cross{G,T}{e}$, and they contribute at least $2$ to $\cng{G,T}{e}$. 
We now have some cases to consider.

If $c$ is a two-literal clause, its root-to-clause edge $(r, c)$ is also in $\cross{G,T}{e}$, by Claim~\ref{claim: root-to-clause edge}. 
Thus, $\cng{G,T}{e}\ge k_2+3 > K$ (see Figure~\ref{fig: claim clause vertices}a). 
So assume now that $c$ is a three-literal clause, and let $\tildex_l \neq \tildex_i, \tildex_j$ be the third literal of $c$.
If $T$ contains $(c,\tildex_l)$, the $x_l$-gadget would also contribute at least $2$ to $\cng{G,T}{e}$, so $\cng{G,T}{e} \ge k_2+4 > K$ (see Figure~\ref{fig: claim clause vertices}b). 
Otherwise, $(c,\tildex_l)\notin T$, and $(c,\tildex_l)$ itself contributes $1$ to $\cng{G,T}{e}$, 
so  $\cng{G,T}{e} \ge k_2+3 > K$ (see Figure~\ref{fig: claim clause vertices}c). 

We have shown that if a clause vertex $c$ is not a leaf in $T$, then in all cases
the congestion of $T$ would exceed $K$, completing the proof of Claim~\ref{claim: clause vertex}.

\begin{figure}[ht]
	\begin{center}
		\includegraphics[height = 1.56in]{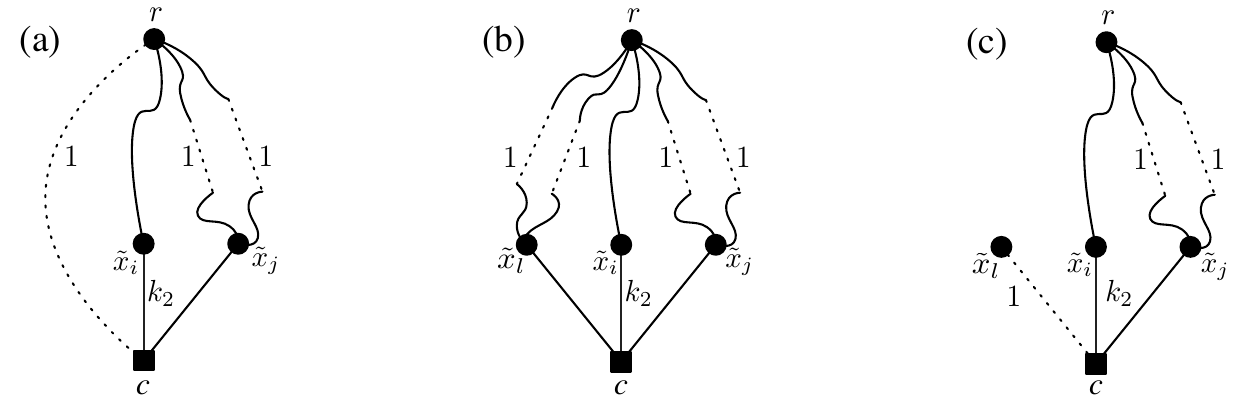}
	\end{center}
	\caption{Illustration of the proof of Claim~\ref{claim: clause vertex}. 
	In~(a) $c$ is a two-literal clause; in~(b) and~(c), $c$ is a three-literal clause.}
	\label{fig: claim clause vertices}
\end{figure}

%

\medskip

\begin{claim}\label{claim: root to xi'}
	For each $x_i$-gadget, edge $(r,x_i')$ is in $T$. 
\end{claim}

\medskip

Towards contradiction, suppose that $(r, x_i')$ is not in $T$.
Let $(x_i',c)$ be the clause-to-literal edge of $x_i'$. 
If only one of the two edges $(x_i',x_i), (x_i',c)$ is in $T$, making $x'_i$ a leaf, then the congestion of that edge is $k_3 + k_2 + 1 > K$. 
Otherwise, both $(x_i',x_i), (x_i',c)$ are in $T$. 
Because $c$ is a leaf in $T$ by Claim~\ref{claim: clause vertex}, $e= (x_i, x_i')$ is the last edge on the $r$-to-$x_i'$ path in $T$. 
As shown in Figure~\ref{fig: claim 4 and 5}a, $\cng{G,T}{e} \ge k_3+k_2+2 > K$. This proves Claim~\ref{claim: root to xi'}.

%

\medskip

\begin{claim}\label{claim: root to y}
	For each $x_i$-gadget, edge $(r,y_i)$ is in $T$. 
\end{claim}

\medskip

To prove this claim, suppose $(r, y_i)$ is not in $T$. 
We consider the congestion of the first edge $e$ on the $r$-to-$y_i$ path in $T$. 
By Claims~\ref{claim: clause vertex} and~\ref{claim: root to xi'}, we have $e = (r, x_i')$, 
all vertices of the $x_i$-gadget have to be in $T_{x_i', r}$, and $T_{x_i', r}$ does not contain literal vertices of another variable $x_j \ne x_i$. 
For each literal $\tildex_i$ of $x_i$,
if a clause-to-literal edge $(c, \tildex_i)$ is in $T$, then the two other edges of $c$ contribute $2$ to $\cng{G,T}{e}$, 
otherwise $(c, \tildex_i)$ contributes $1$ to $\cng{G,T}{e}$. 
Then, $\cng{G,T}{e} \ge k_4+k_3+3 > K$ (see Figure~\ref{fig: claim 4 and 5}b), proving Claim~\ref{claim: root to y}.

\begin{figure}[ht]
	\begin{center}
		\includegraphics[height = 1.8in]{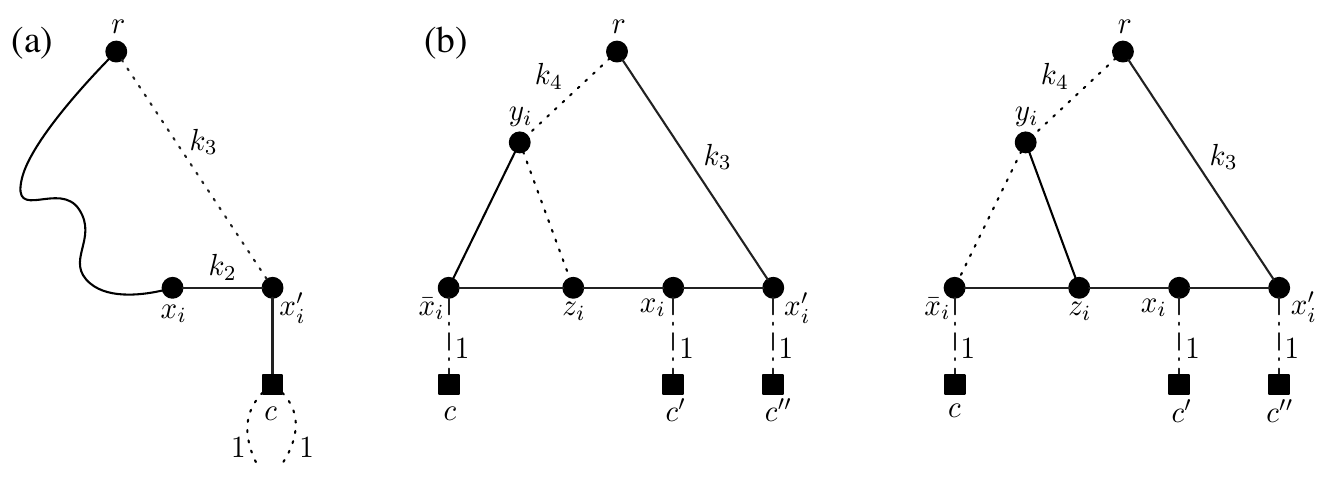}
	\end{center}
	\caption{(a) Illustration of the proof of Claim~\ref{claim: root to xi'}. 
	(b) Illustration of the proof of Claim~\ref{claim:  root to y}. Dot-dashed lines are edges that may or may not be in $T$.}
	\label{fig: claim 4 and 5}
\end{figure}

%

\medskip

\begin{claim} \label{claim: z to x}
	For each $x_i$-gadget, exactly one of edges $(z_i, x_i)$ and $(x_i, x_i')$ is in $T$. 
\end{claim}

\medskip

By Claims~\ref{claim: root to xi'} and \ref{claim: root to y}, edges $(r,y_i)$ and $(r,x_i')$ are in $T$. 
Since the clause neighbor $c'$ of $x_i$ is a leaf of $T$, by Claim~\ref{claim: clause vertex}, if none of $(z_i, x_i)$, $(x_i, x_i')$ were in $T$, $x_i$ would not be reachable from $r$ in $T$. 
Thus, at least one of them is in $T$.
Now, assume both $(z_i, x_i)$ and $(x_i, x_i')$ are in $T$  (see Figure~\ref{fig: claim 6 and 7}a). 
Then, edge $(y_i, z_i)$ is not in $T$, as otherwise we would create a cycle. 
Let us consider the congestion of edge $e = (r,x_i')$. 
Clearly, $x_i$ and $x_i'$ are in $T_{x_i', r}$. 
The edges of the two clause neighbors $c'$ and $c''$ of $x_i$ and $x_i'$ contribute at least $2$ to $\cng{G,T}{e}$, by Claim~\ref{claim: clause vertex}. 
In addition, by Claim~\ref{claim: root to literal path}, besides $e$ and $(y_i, z_i)$, $\cross{G,T}{e}$ contains another edge of the $x_i$-gadget which contributes at least another $1$ to $\cng{G,T}{e}$. 
Thus, $\cng{G,T}{e} \ge k_4+k_3+3 > K$ --- a contradiction. 
This proves Claim~\ref{claim: z to x}.

%

\medskip

\begin{claim}\label{claim: y to z}
	For each $x_i$-gadget, edge $(y_i, z_i)$ is in $T$. 
\end{claim}

\medskip

By Claims~\ref{claim: root to xi'} and \ref{claim: root to y}, the two edges $(r,x_i')$ and $(r,y_i)$ are in $T$. 
Now assume, towards contradiction, that $(y_i, z_i)$ is not in $T$ (see Figure~\ref{fig: claim 6 and 7}b). 
By Claim~\ref{claim: z to x}, only one of $(z_i, x_i)$ and $(x_i, x_i')$ is in $T$. 
Furthermore, the clause neighbor $c'$ of $x_i$ is a leaf of $T$, by Claim~\ref{claim: clause vertex}.
As a result, $(z_i, x_i)$ cannot be on the $y_i$-to-$z_i$ path in $T$. 
To reach $z_i$ from $y_i$, the two edges  $(y_i, \barx_i), (\barx_i, z_i)$ have to be in $T$. 
Let us consider the congestion of $e = (y_i, \barx_i)$. 
The edges of the clause neighbor $c$ of $\barx_i$ contribute at least $1$ to the congestion of $e$, by Claim~\ref{claim: clause vertex}. 
Also, by Claim~\ref{claim: root to literal path}, besides $e$ and $(y_i, z_i)$, $\cross{G,T}{e}$ contains another edge of the $x_i$-gadget which contributes at least $1$ to $\cng{G,T}{e}$. 
In total, $\cng{G,T}{e} \ge k_4+k_2+2 > K$, reaching a contradiction and completing the proof of Claim~\ref{claim: y to z}.

\begin{figure}[ht]
	\begin{center}
		\includegraphics[height = 1.6in]{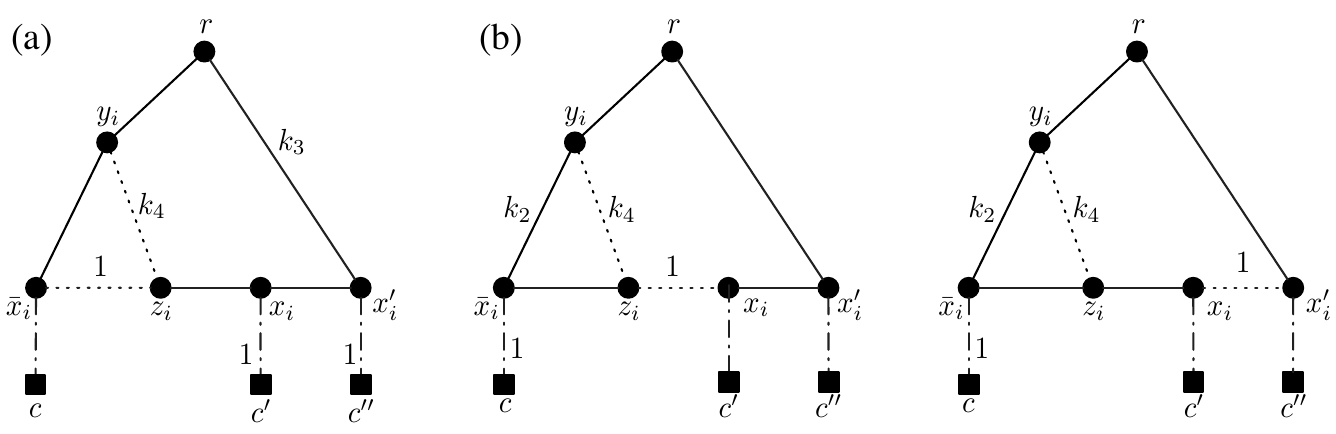}
	\end{center}
	\caption{(a) Illustration of the proof of Claim~\ref{claim: z to x}. (b) Illustration of the proof of Claim~\ref{claim: y to z}.}
	\label{fig: claim 6 and 7}
\end{figure}
%

\medskip

\begin{claim}\label{claim: conflict assignment}
	For each $x_i$-gadget, if its clause-to-literal edge $(\barx_i, c)$ is in $T$, then its other two clause-to-literal 
	edges $(x_i, c')$ and $(x_i', c'')$ are not in $T$. 
\end{claim}

\medskip

Assume the clause-to-literal edge $(\barx_i, c)$ of the $x_i$-gadget is in $T$. 
By Claim~\ref{claim: y to z}, edge $(y_i, z_i)$ is in $T$. 
If $(y_i, \barx_i)$ is also in $T$, edge $(\barx_i, z_i)$ cannot be in $T$, and it contributes $1$ to $\cng{G,T}{y_i, \barx_i}$. 
As shown in Figure~\ref{fig: claim5}a, $\cng{G,T}{y_i, \barx_i}  = k_2 + 3 > K$. 
Thus, $(y_i, \barx_i)$  cannot be in $T$. 
Since $c$ is a leaf of $T$, edge $(\barx_i, z_i)$ has to be in $T$, for otherwise $\barx_i$ would not be reachable from $r$. 
By Claim~\ref{claim: z to x}, one of edges $(z_i, x_i)$ and $(x_i, x_i')$ is in $T$. 
If $(z_i, x_i)$ is in $T$  (see Figure~\ref{fig: claim5}b), $\cng{G,T}{y_i, z_i} \ge k_4 + 5 > K$. 
Hence, $(z_i, x_i)$ is not in $T$, which implies that $(x_i, x_i')$ is in $T$. 

\begin{figure}[ht]
	\begin{center}
		\includegraphics[height = 1.8in]{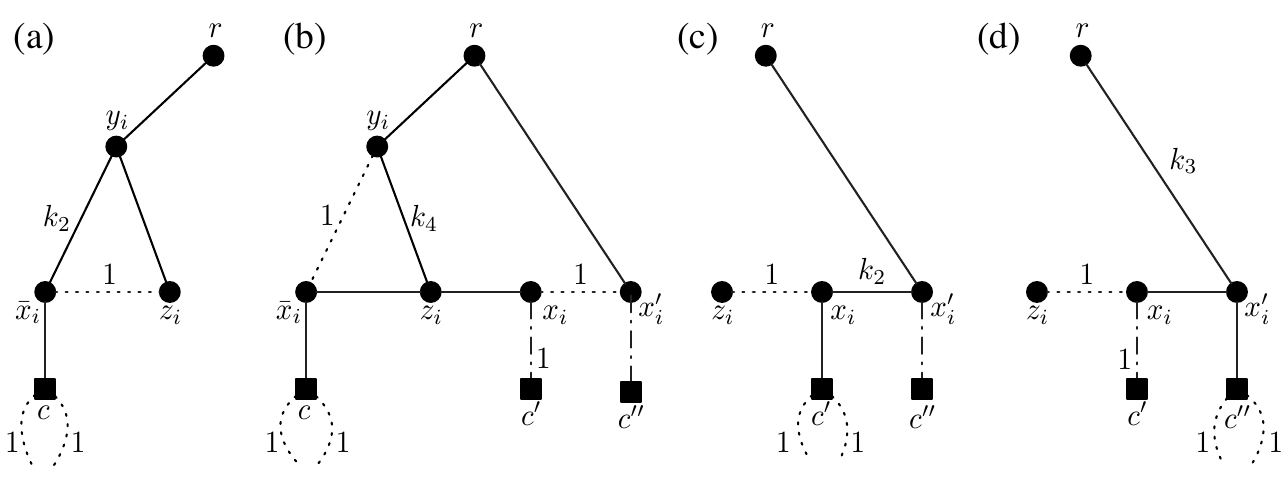}
	\end{center}
	\caption{Illustration of the proof of Claim~\ref{claim: conflict assignment}.}
\label{fig: claim5}
\end{figure}

\smallskip
Now, we proceed by contradiction assuming that at least one other clause-to-literal edge of the $x_i$-gadget is in $T$. 
If edge $(x_i, c')$ is in $T$, $\cng{G,T}{x_i, x_i'} \ge k_2 + 3 > K$, as shown in Figure~\ref{fig: claim5}c. 
Similarly, if $(x_i', c'')$ is in $T$, $\cng{G,T}{r, x_i'} \ge k_3 + 4 > K$ (see Figure~\ref{fig: claim5}d). 
So we reach a contradiction in both cases, thus proving Claim~\ref{claim: conflict assignment}. 

\medskip
We are now ready to complete the proof of the $(\Leftarrow)$ implication in the equivalence $(\ast)$.
We use our spanning tree $T$ of congestion at most $K$ to
create a truth assignment for $\phi$ by setting $x_i = 0$ if the clause-to-literal edge of $\barx_i$ is in $T$, otherwise $x_i = 1$. 
By Claim~\ref{claim: conflict assignment}, this truth assignment is well-defined.
Each clause has one clause-to-literal edge in $T$ which ensures that all clauses are indeed satisfied. 



\section{$\NP$-completeness proof of $\problemKSTC{K}$ for bipartite graphs of radius $2$ and $K \ge 6$}
\label{sec: k-stc hardness bipartite graph radius 2}
In this section we establish the following result:

\smallskip

\begin{theorem}
\label{theorem: kstc bipartite graph radius 2}
For any fixed integer $K \ge 6$, $\problemKSTC{K}$ is $\NP$-complete for bipartite graphs of radius $2$,
even if they have only one vertex of degree greater than $\max(6, K-2)$. 
\end{theorem}

\smallskip
First, we introduce a restricted variant of the satisfiability problem, which we name (M2P1N)-SAT, that is used in the reduction. An instance of (M2P1N)-SAT
is a boolean expression in conjunctive normal form with the following properties:   
\begin{itemize}[label=$\circ$]
	\item Each clause either contains three positive literals (a 3P-clause), or two positive literals (a 2P-clause), or two negative literals (a 2N-clause). 
		Also, literals in the same clause are of different variables. 
	\item Each variable appears exactly three times: once in a 3P-clause, once in a 2P-clause and once in a 2N-clause. 
	\item Two clauses share at most one variable. 
\end{itemize}

\smallskip

\begin{lemma}
	\label{lemma: M2P1N-SAT}
	 (M2P1N)-SAT is $\NP$-complete. 
\end{lemma}

\begin{proof}
It is clear that (M2P1N)-SAT belongs to $\NP$. 
To demonstrate $\NP$-completeness, we show a polynomial-time reduction from the $\NP$-complete problem
called BALANCED-3SAT~\cite{hagele_2001_complexity_scheduling_commercials}. BALANCED-3SAT is a restriction
of the satisfiability problem to boolean expressions in conjunctive normal form where, for each variable $x$,
the positive literal $x$ appears the same number of times as the negative literal $\barx$.
We can further assume that every variable appears at least four times, and that, for each clause, all variables
that appear in this clause are different. 

Given an instance $\psi$ of BALANCED-3SAT, we construct an instance $\phi$ of (M2P1N)-SAT as follows:
 
\begin{itemize}[label=$\circ$]
	\item For each variable $x$ in $\psi$, if $x$ appears $2t$ times (for some integer $t \ge 2$),  
			create $2t$ new variables $x_0, x_1, \ldots, x_{2t-1}$.
	\item Replace the $t$ positive occurrences of $x$ by even-indexed variables $x_0, x_2, \ldots, x_{2t-2}$, 
	and replace its $t$ negative occurrences by odd-indexed variables $x_1, x_3, \ldots, x_{2t-1}$. 
	\item Add $t$ clauses of the form $(x_i \lor x_{i+1})$ for $i = 0, 2, \ldots, 2t-2$, 
	and $t$ clauses of the form $(\barx_{i} \lor \barx_{(i+1) \!\! \mod 2t})$ for $i = 1, 3, \ldots, 2t-1$. 
\end{itemize} 

\smallskip
By the construction,  $\phi$ is a correct instance of (M2P1N)-SAT. 
For each variable $x$ of $\psi$, its corresponding ``cycle'' of the newly added two-literal clauses in $\phi$ ensures that 
$x_0 = \barx_1 = x_2 = \barx_3 = \ldots = x_{2t-2} = \barx_{2t-1}$. 
Thus, a truth assignment that satisfies $\psi$ can be converted into a truth assignment that satisfies $\phi$ by setting the even-indexed variables 
to the truth value of the original variable in $\psi$, and the odd-indexed variables to the opposite value.
Conversely, a truth assignment that satisfies $\phi$ can be converted into a truth assignment that satisfies $\psi$
by reversing this process. This shows
that $\psi$ is satisfiable if and only if $\phi$ is satisfiable, completing the proof of the lemma. 
\end{proof}


In order to prove Theorem~\ref{theorem: kstc bipartite graph radius 2}, we show a polynomial-time reduction from (M2P1N)-SAT. 
Given an instance $\phi$ of  (M2P1N)-SAT, we construct a graph $G$ such that

\smallskip
\begin{description}
	\item{$(\ast)$} $\phi$ has a satisfying truth assignment if and only if $\stc{G} \le K$. 
\end{description}

\smallskip
Graph $G$ will be bipartite, of radius $2$, and  will have only one vertex of degree larger than $\max(6, K-2)$. We will
describe $G$ using some double-weighted edges, that we refer to as \emph{fat edges}.
As previously discussed in Section~\ref{sec: preliminaries}, here we need a specific realization
of these double weighted edges, in which weights are realized using the spintop graph.
(See Figures~\ref{fig: single weight} and~\ref{fig: double weights}.)
For $i\in\braced{1,2,3,4,5}$, let $k_i = K- i$.
We start with an empty graph $G$ and proceed as follows:

\begin{itemize}[label=$\circ$]
	\item  Add a \emph{root vertex} $r$. 
	\item  For each variable $x$ of $\phi$, add a \emph{variable vertex} $x$ and a \emph{root-to-vertex} edge $(r, x)$.
	\item  For each clause $c$, add a clause vertex $c$, and add edges from $c$ to the 
			vertices representing variables whose literals (positive or negative) appear in $c$. 
	If clause $c$ contains all positive literals, we call its clause-to-variable edges \emph{positive-clause edges}, 
					otherwise its clause-to-variable edges are \emph{negative-clause edges}. 
	\item For each 2P-clause vertex $c$, add a fat edge $(r,c)$ of double weight $\dw{k_5}{k_1}$.
	\item For each 2N-clause vertex $c$, add a fat edge $(r,c)$ of double weight $\dw{k_4}{k_1}$.
\end{itemize} 

See Figure~\ref{fig: radius 2 graph and tree}a for an example of a partial graph constructed using the above rules.  
By routine inspection, taking into account that the weighted edges use the  spintop realization, $G$ 
is bipartite, all vertices are at distance at most $2$ from $r$, and $r$ is the only vertex of degree larger than $\max(6, K-2)$.
We now proceed to show that $G$ satisfies property~$(\ast)$.

\begin{figure}[ht]
	\begin{center}
	\includegraphics[height = 1.7in]{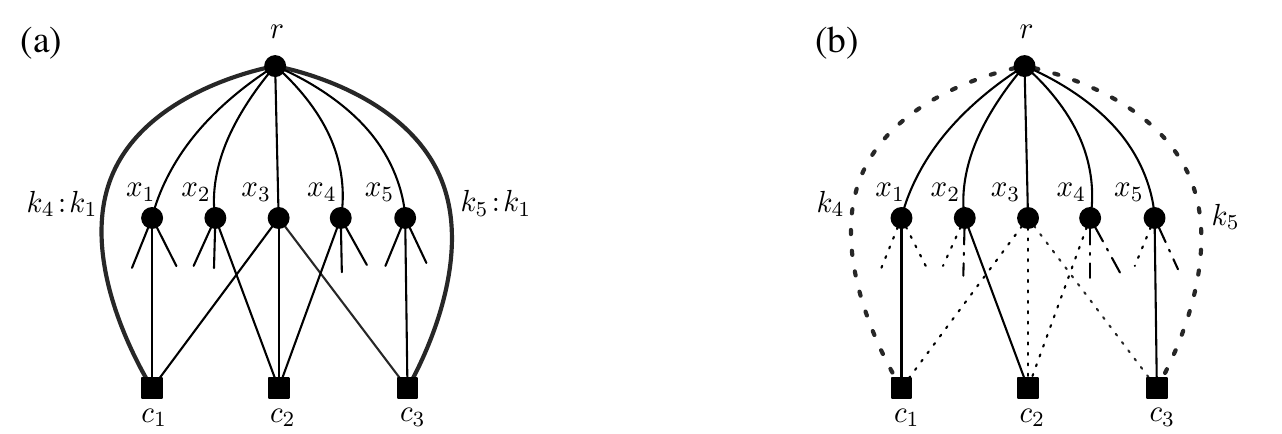}
	\end{center}
	\caption{(a) An example of a partial graph $G$ for $\phi = c_1 \land c_2 \land c_3 \land \cdots$ where 
	$c_1 =(\bar{x}_1 \lor \bar{x}_3), c_2 = (x_2 \lor x_3 \lor x_4), c_3 = (x_3 \lor x_5)$. Bold lines represent fat edges with given double weights.
	(b) An example of a partial tree $T$ of $G$ where $x_1$ is chosen by $c_1$, $x_2$ by $c_2$, $x_5$ is by $c_3$.
	Solid lines are tree edges, dotted lines are non-tree edges, and dot-dashed lines are edges that may or may not be in $T$. 
	Non-tree double-weighted edges contribute the indicated weights to edge congestion.}
	\label{fig: radius 2 graph and tree} 
\end{figure}

\medskip
\noindent
$(\Rightarrow)$ Assume that $\phi$ has a satisfying truth assignment.
From this assignment we construct a spanning tree $T$ of $G$ by adding all root-to-vertex edges, and,
for each clause $c$, adding to $T$ an edge from $c$ to any variable vertex whose literal satisfies $c$ (see Figure~\ref{fig: radius 2 graph and tree}b).  
By the construction, $T$ is a spanning tree of $G$. Note that all clause vertices in $T$ are leaves and all fat edges are non-tree edges.

Now, we proceed to verify that all tree edges of $T$ have congestion at most $K$. 
We start with leaf edges of $T$. The congestion of the leaf edge of a 3P-clause is $3$.
For a 2P-clause, the congestion of its leaf edge is $K-3$, because its fat edge contributes $k_5 = K-5$. 
For a 2N-clause, the congestion of its leaf edge is $K-2$, because its fat edge contributes $k_4 = K-4$. 

Next, consider the root-to-vertex edge of a variable $x_i$.
If $x_i$ is not chosen to satisfy any clauses, then $\cng{G,T}{r,x_i} = 4$  (see  Figure~\ref{fig: tree bipartite graphs}a). 
If it is chosen to satisfy only its 3P-clause, then $\cng{G,T}{r,x_i} = 5$  (see  Figure~\ref{fig: tree bipartite graphs}b). 
If it is chosen to satisfy only its 2P-clause, then $\cng{G,T}{r,x_i} = k_5 + 4 = K-1$  (see  Figure~\ref{fig: tree bipartite graphs}c). 
If it is chosen to satisfy both its 3P-clause and its 2P-clause, then $\cng{G,T}{r,x_i} = k_5 + 5  = K$  (see  Figure~\ref{fig: tree bipartite graphs}d). 
Finally, if it is chosen to satisfy its 2N-clause, then $\cng{G,T}{r,x_i} = k_4 + 4 = K$  (see  Figure~\ref{fig: tree bipartite graphs}e). 

There are also edges inside the realizations of fat edges, but their congestion does not exceed $K$, by Lemma~\ref{lemma: double weights}.
We have thus shown that the congestions of all edges in $T$ are at most $K$; that is, $\stc{G} \le K$.

\begin{figure}[ht]
	\begin{center}
	\includegraphics[height = 1.3in]{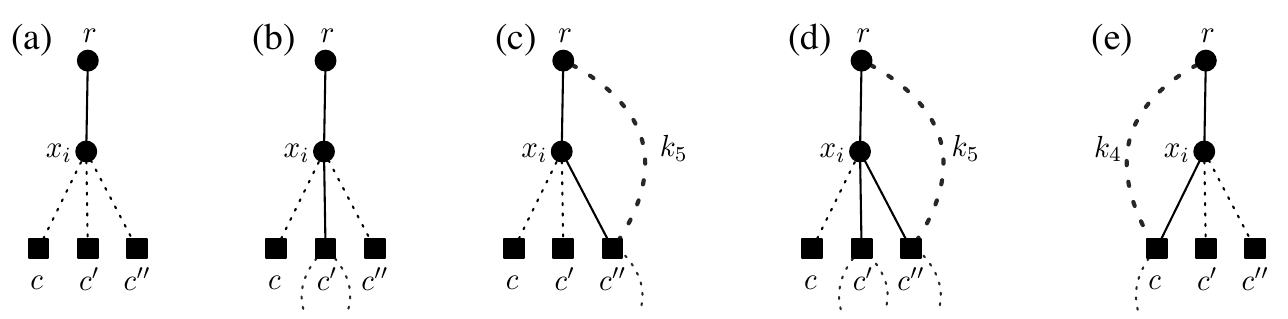}
	\end{center}
	\caption{By $c, c', c''$, we denote the 2N-clause, 3P-clause and 2P-clause of $x_i$ respectively. 
	In (a), $x_i$ is not chosen by any clause, it is chosen by $c'$ in (b), by $c''$ in (c), by both $c'$ and $c''$ in (d), and by $c$ in (e).}
	\label{fig: tree bipartite graphs}
\end{figure}

\smallskip
$(\Leftarrow)$ Assume $T$ is a spanning tree of $G$ with $\cng{G}{T} \le K$. 
From $T$, we will construct a satisfying truth assignment for $\phi$. The argument here, while much shorter, has a subtle
aspect that was not present in the proof in Section~\ref{sec: k-stc hardness proof}, namely
now it is not necessarily true that all clause vertices in $T$ are leaves. (It's not hard to see that for large $K$
a single branch out of $r$ may visit multiple variables via their 3P-clause vertices.)

We present two claims showing that $T$ must have a special form that will allow us to define the truth assignment for $\phi$. 


\smallskip

\begin{claim}
\label{claim: no fat edge}
For each two-literal clause $c$, its fat edge $(r,c)$ is not in $T$. 
\end{claim}

\smallskip

For each literal of $c$, there is an $r$-to-$c$ path via the variable vertex of this literal. So, 
together with edge $(r,c)$, $G$ has three disjoint $r$-to-$c$ paths. 
Thus, if $(r,c)$ were in $T$, its congestion would be at least $k_1+2>K$, proving Claim~\ref{claim: no fat edge}.


\smallskip

\begin{claim}
\label{claim: no conflict signs}
For each variable vertex $x_i$, if its negative-clause edge is in $T$ then its two positive-clause edges are not in $T$. 
\end{claim}

\smallskip

Denote by $c, c', c''$ the 2N, 3P, 2P-clause vertices of $x_i$ respectively. 
Since $c,c',c''$ all contain variable $x_i$, they cannot share any other variables (by the definition of (M2P1N)-SAT).
Therefore, the four literals in $c,c',c''$ other than $x_i$ and $\barx_i$ must all involve different variables. 

Toward contradiction, suppose $(x_i, c)$ and at least one of $(x_i, c'), (x_i, c'')$ are in $T$. 
We will estimate the congestion of the first edge $e = (r,v)$ on the $r$-to-$c$ path in $T$. 

By Claim~\ref{claim: no fat edge}, fat edge $(r,c)$ contributes $k_4$ to $\cng{G,T}{e}$.  
The rest of the argument is based on the following two observations: (i) If a clause $\tildec \in \braced{c, c',c''}$ is in $T_{v,r}$, 
and some variable $x$ is in $\tildec$, then either $(r,x)$ or $(x, \tildec)$ is in $\cross{G,T}{e}$; that is, this $x$ contributes $1$ to $\cng{G,T}{e}$.
(This is true whether or not $v = x$. And if $x = x_i$ and $\tildec = c$, then $(r,x_i)$ is the edge that contributes to $\cng{G,T}{e}$.)
On the other hand, (ii) if a clause $\tildec \in \braced{c', c''}$ is not in $T_{v,r}$, then $(x_i,\tildec)$ contributes $1$ to $\cng{G,T}{e}$.

Now we have some cases to consider. First, if $c' \in T_{v,r}$ and $c'' \notin T_{v,r}$, by the above observations, 
four different variables in $c, c'$ contribute $4$ to $\cng{G,T}{e}$ and $(x_i, c'')$ contributes $1$. In total, $\cng{G,T}{e} \ge k_4 + 4+1 > K$. 
On the other hand, when $c'' \in T_{v,r}$ and $c' \notin T_{v,r}$, the three different variables of $c, c''$ contribute $3$ while $(x_i, c')$ 
contributes $1$ to $\cng{G,T}{e}$.  
Also, the fat edge $(r,c'')$ contributes $k_5$, by Claim~\ref{claim: no fat edge}. Thus, $\cng{G,T}{e} \ge k_4 + k_5 + 3+1 > K$. 
Lastly, when both $c', c''$ are in $T_{v,r}$, the five different variables of $c, c', c''$ contribute to $\cng{G,T}{e}$,
so  $\cng{G,T}{e} \ge k_4 + 5 > K$. 
We have thus shown that the congestion of $e$ exceeds $K$ in all cases, completing the proof of Claim~\ref{claim: no conflict signs}.

\smallskip
We are now ready to describe the truth assignment for $\phi$ using $T$. 
For each variable $x_i$, assign $x_i = 0$ if its negative clause edge is in $T$, otherwise, $x_i=1$. 
By Claim~\ref{claim: no conflict signs}, the truth assignment is well-defined. 
By Claim~\ref{claim: no fat edge}, each clause vertex has at least one edge to a variable vertex, which ensures all clauses are satisfied. 
This completes the proof of Theorem~\ref{theorem: kstc bipartite graph radius 2}.


\section{Complexity results of $\problemKDSTC{K}{2}$}
\label{sec: k-dstc-2}

In this section, we consider problem $\problemKDSTC{K}{D}$ where, given a graph $G$, the objective is to determine
if $G$ has a depth-$D$ spanning tree of congestion at most $K$.
Here, as before, $K$ is a fixed positive integer.  We present the following results: 

\begin{theorem}
\label{thm: K2STC np-hardness}
For any fixed integer $K \ge 6$, $\problemKDSTC{K}{2}$ is  $\NP$-complete for bipartite graphs, even if they have only one vertex of degree greater than $\max(6, K-2)$. 
\end{theorem}

\begin{theorem}
\label{thm: K2STC algorithm}
For any fixed integer $K\le 5$, $\problemKDSTC{K}{2}$ is polynomial-time solvable for bipartite graphs. 
\end{theorem}

We remark that the complexity status of $\problemKDSTC{K}{2}$ is independent of whether the root of the spanning tree is specified
or not, because there are at most $n$ choices for $r$. This establish the equivalence of these two versions (with or without the root specified)
in terms of polynomial-time solvability or $\NP$-hardness.


\subsection{$\NP$-completeness proof of $\problemKDSTC{K}{2}$ for $K \ge 6$}
\label{subsec: NP-complete K2STC}

The proof of Theorem~\ref{thm: K2STC np-hardness} can be easily derived from the proof of Theorem~\ref{theorem: kstc bipartite graph radius 2} 
in Section~\ref{sec: k-stc hardness bipartite graph radius 2}.  The reduction remains unchanged.
In that construction, the bipartite partition of $G$ has two parts: $X$, which includes vertices adjacent to the root $r$ 
(the variable vertices and parts of the spintop gadgets), and 
$C \cup \braced{r}$, which includes the remaining vertices (the clause vertices, the root, and 
the vertices not adjacent to $r$ in the spintop gadgets).  
The proof for the forward direction is also identical, 
since the depth of the spanning tree generated from the proposed construction is already two.
 
For the reverse implication, suppose $T$ is the depth-two spanning tree with congestion at most $K$. 
We present a simple claim about the structure of $T$:

\begin{claim}
\label{claim: edges to first layer}
All edges incident to $r$ are in $T$, and all vertices in $C$ are leaves of $T$. 
\end{claim}

Since $G$ does not have any eccentricity-one vertex and the only vertex in $G$ of eccentricity two is $r$, $T$ has to be rooted at $r$, 
which implies that the paths from $r$ to other vertices in $T$ have length at most 2. 
If an edge $(r, x) \in G$ were not in $T$, the $r$-to-$x$ path in $T$ would have length at least 3, which is a contradiction. 
Thus, $T$ traverses all edges of $r$. 
The second part of the claim follows directly from the first part. 

In addition to Claim~\ref{claim: edges to first layer}, $T$ also has the two properties described in Claim~\ref{claim: no fat edge} (which can be established using the same argument) 
and Claim~\ref{claim: no conflict signs} (its proof can be made simpler by considering the fact about clause vertices being leaves of $T$).

Finally, the truth assignment for $\phi$ can be created the same way as in  Section~\ref{sec: k-stc hardness bipartite graph radius 2}.


\subsection{An algorithm for $\problemKDSTC{K}{2}$ in bipartite graphs for $K\le 5$}
\label{subsec: polynomial 5STC2}

We now prove Theorem~\ref{thm: K2STC algorithm}. We only give an explicit algorithm for $K=5$. This is because
$\problemKDSTC{K}{2}$ is trivial for $K=1$, and for $K=2$, the problem can be solved by a straighforward adaptation
of the algorithm in~\cite{otachi_2010_complexity_result_stc}, even for general graphs. 
The cases when $K=3,4$ can be handled by slightly modifying (in fact, simplifying) the algorithm for $K=5$ below.
(Alternatively, for $K=3$, one can adapt the algorithm from~\cite{otachi_2010_complexity_result_stc}.) 

So let's assume that $K=5$ and let $G$ be a given bipartite graph.
If $\radius{G} > 2$, then $G$ does not have any spanning tree of depth two. 
If $\radius{G} = 1$, then $G$ must be a complete bipartite graph where one partition contains only one vertex, that is $G$ itself is a tree of depth one and its congestion is one.  
Thus, we can assume $\radius{G} = 2$, which means that any depth-two spanning tree of $G$ has to be rooted at a vertex with eccentricity two. 
There are at most $n$ such vertices, and for each we can check, using the procedure described below,
whether there is a depth-two spanning tree $T$ rooted at $r$ such that $\cng{G}{T} \le 5$.  
Therefore from now on we can assume that this $r$ is already given.

Let $X$ and $C \cup \braced{r}$ be the two parts of the bipartition of $G$. 
Let $E_r$ be the set of edges incident to $r$, and $E_s = E \setminus E_r$. 
We can make the following assumptions (that can be implemented in a pre-processing stage):
\begin{itemize}
\item
We can assume that all vertices in $G$ have degree at least 2, since removing (repeatedly) degree-one vertices does not affect the spanning tree congestion of the graph. 
\item
By Claim~\ref{claim: edges to first layer}, each vertex $c \in C$ has to be a leaf in any depth-two spanning tree rooted at $r$, 
and the congestion of its leaf edge is equal to $\degree{G}{c}$.  
Thus, we can also assume that $\degree{G}{c}\le 5$ for all $c \in C$.
\item
Similarly, each edge $(r,x)$ must be in a spanning tree of depth two. 
With the assumptions above, each edge $(x,c)$ from $x$ to $c \in C$ contributes to the congestion of $(r,x)$, either directly, 
if it's not in the tree, or indirectly, if it's in the tree (as then the other edges from this $c$ contribute, and there is at least one). 
Therefore, if $\degree{G}{x} > 5$ for some $x \in X$, we would have $\cng{G,T}{r,x} > 5$. 
So we can assume that $\degree{G}{x} \le 5$ for all $x\in X$.
\end{itemize}


\myparagraph{Algorithm outline} 
The general idea of the algorithm is to start with a tree $T$ that contains only edges in $E_r$ and gradually add leaf edges for all vertices  $c \in C$. 
This can be naturally interpreted as assigning vertices in $C$ to vertices in $X$.
If $c \in C$ and $x \in \neighbor{G}{c}$, then assigning $c$ to $x$ means that edge $(c,x)$ is being added to $T$.
If it is possible to assign all vertices in $C$ to some vertices in $X$, while ensuring that the congestions of the edges in $E_r$ do not exceed 5, 
then $T$ will be the desired spanning tree. 
In the first phase, we will do this assignment one vertex at a time. 
Call the assignment $c\to x$ \emph{feasible} if it does not cause the current congestion of $(r,x)$ to exceed $5$. 
Such a feasible assignment can be made safely if it either is forced (say, if $c$ can be assigned to only one vertex in $X$ without exceeding the congestion bound), 
or it can be made without loss of generality (that is, if we can show that if
there is any spanning tree with congestion at most $5$, then there is also one that makes this specific assignment). 
To achieve this, we will carefully track the congestion of the edges in $E_r$ throughout the construction.  
The first phase will end with all yet unassigned vertices in $C$ of degree $3$ or $4$. Then the only way to complete 
the assignments is by adding a matching between $C$ and $X$, and this is done in the second phase.


\myparagraph{Phase~1}
Initially $T$ contains only the edges from $r$ to $X$. 
During the process, besides these edges, $T$ will also contain one edge $(c,x)$ for each $c \in C$ that is already assigned to $x \in \neighbor{G}{c}$. 
For this (not yet spanning) tree $T$, define the congestion of a vertex $x \in X$ in the current stage of $T$ as: 
\begin{equation}
\label{eqn: cng partial tree}
	\cng{}{r,x} \;=\; \degree{G}{x} + \sum_{c\to x} [\degree{G}{c} - 2]
\end{equation}
where the sum is over all $c\in C$ that are assigned to $x$. Thus,
when a vertex $c \in C$ get assigned to a vertex $x \in \neighbor{G}{c}$, the congestion of $(r, x)$ increases by $\degree{G}{c} - 2 \ge 0$. 
Note that after this assignment, $\cng{}{e}$ remains unchanged for $e \in E_r \setminus \braced{(r, x)}$ and the 
congestions of $(r, x)$ is non-decreasing. 


\emph{Assigning degree-$2$ vertices.}
For a vertex $c$ of degree 2, let $(x,c)$ be any of its edges, and assign $c$ to $x$.
The congestion of $(r,x)$ remains unchanged. 


\emph{Assigning degree-$5$ vertices.}
For a vertex $c$ of degree $5$, if we assign $c$ to a vertex $x$, the congestion of $(r,x)$ would increase by $3$. 
Therefore, $c$ can only be assigned to $x$ if the congestion of $(r,x)$ is $2$ prior to the assignment,
which implies that the only edge in $E_s$ that is incident to $x$ is $(c,x)$. 
Including $(c,x)$ in $T$ would not affect the congestion of $(r,x)$ in subsequent steps, as $c$ is the
only vertex in $C$ that can be assigned to $x$.
If there is no $x$ that satisfies the requirement, we terminate and report failure.
If there are multiple feasible choices for such $x$, we can choose any of them. 
This is valid, because if $x'\in X$ is another candidate, then $x'$ will not be assigned to any vertices in $C$ and the  congestion of $(r, x')$ will remain $2$.

\emph{Assigning pairs of degree-$3$ vertices to the same vertex.}
If there are two degree-$3$ vertices $c_1, c_2  \in C$ that share the same neighbor $x$, and $\neighbor{G}{x} = \braced{r, c_1, c_2}$, we can assign both $c_1$ and $c_2$ to $x$. 
The congestion of $(r,x)$ will increase to $5$, and it will remain $5$ since $x$ cannot be assigned to any other vertices in $C$. 
Similar to the previous step, if there is more than one such choice of $x$, any option is valid. 

\smallskip
\myparagraph{Phase~2}
After the first phase, we denote by $C'$ the set of yet unassigned vertices in $C$. 
The vertices in $C'$ have degree either $3$ or $4$. 
Unlike the previous phase, assignments for vertices in $C'$ cannot be made independently. 
We observe that  each of these vertices must be assigned to a different vertex in $X$ because assigning two or more of them to 
the same $x$ would cause the congestion of $(r,x)$ to exceed 5. 
(This is because after Phase~1, if two vertices in $C'$ share a neighbor in $X$ then they cannot both have degree $3$.)
Based on this observation, we can assume that $|X|\ge |C'|$ -- if not, we can report that the congestion is larger than $5$.
Then an assignment of all vertices in $C'$ forms a perfect matching between $C'$
and $X$, that is, a matching that covers all vertices in $C$ (but not necessarily in $X$). Our goal now is to find this matching.

Towards this end, we consider a bipartite subgraph $G'$ of $G$ where one partition consists of the vertices of $C'$, the other 
partition consists of the vertices in $X$, and an edge between $c\in C$ and $x\in X$ is included in $G'$ iff 
$x \to c$ is a feasible assignment. 
We then determine,  in polynomial-time~\cite{1973_hopcroft_algorithm_for_maximum_matchings}, whether $G'$
has a perfect matching.
This matching will define the assignments for vertices in $C'$, ensuring that after all assignments are made, the resulting $T$ is now a spanning tree with congestion at most 5. 
If there is no perfect matching, we report failure. 

%
%
%


\section{Polynomial-time solvability of $\problemDSTC{2}$ in bipartite graphs with vertex degree restrictions}
\label{sec: dstc-2}
Building upon Section~\ref{sec: k-dstc-2}, we continue to explore the variant of $\problemDSTC{2}$, which involves finding a depth-2 spanning tree with minimum congestion in bipartite graphs. 
We provide two polynomial-time algorithms for cases when vertex degrees are restricted:

\begin{theorem}
\label{thm: STC2 algorithm for 3-regular X}
$\problemDSTC{2}$ can be solved in polynomial time when all vertices in $X$ have degree at most 3. 
\end{theorem}

\begin{theorem}
\label{thm: STC2 algorithm for regular C}
$\problemDSTC{2}$ can be solved in polynomial time when all degrees in $C$ have the same degree.
\end{theorem}

To prove each theorem, given any positive integer $K$, we provide an algorithm to construct a depth-2 spanning tree $T$ with congestion at most $K$ (if such a tree exists). 
This implies the polynomial-time solvability of $\problemDSTC{2}$ in these cases. 
The proofs are given in Sections~\ref{subsec: STC2 matching} and~\ref{subsec: STC2 flow}, respectively.

We use the same notation and terminology as in Section~\ref{subsec: polynomial 5STC2}, and we adopt, without loss of 
generality, similar simplifying assumptions.  Let $G$ be the given bipartite graph.
We can assume that $\radius{G} = 2$ and the root $r$ of the desired spanning tree is given. 
We use $X$ and $C \cup \braced{r}$ to refer to the two partitions of the vertices of $G$, and $E_r$ to refer to the set of edges incident to $r$. 

Using the results described in Section~\ref{subsec: polynomial 5STC2}, we can solve $\problemKDSTC{K}{2}$ for $K \le 5$.  
Thus, we will assume $K \ge 6$. Also, as in Section~\ref{subsec: polynomial 5STC2},
we can assume that $2 \le \degree{G}{v} \le K$ for any $v \in C \cup X$ and $\degree{G}{r} \ge 2$. 

Both algorithms start with a tree $T$ that contains only edges in $E_r$. 
The goal is adding leaf edges for all vertices in $C$ while ensuring that the congestion of edges in $E_r$ does not exceed $K$. 
For a vertex $x \in X$, the congestion of edge $(r,x)$ in $T$ is defined in the same way as in Equation~\ref{eqn: cng partial tree}.


\subsection{$\problemKDSTC{K}{2}$ for bipartite graphs with all degrees in $X$ at most 3}
\label{subsec: STC2 matching}

We now present the proof of Theorem~\ref{thm: STC2 algorithm for 3-regular X}, namely a polynomial-time algorithm for
$\problemKDSTC{K}{2}$ restricted to bipartite graphs $G$ where the degree of the vertices in $X$ is at most 3. 
The general idea of this algorithm is similar to the $\problemKSTC{5}{2}$ algorithm described in Section~\ref{subsec: polynomial 5STC2}. 
The process consists of two phases: in the first phase we create assignments for vertices in $C$ that are adjacent to degree-2 vertices  in $X$.
Then, in the second phase, the remaining assignments are determined by a perfect matching in an auxiliary graph $H$ constructed in polynomial time from $G$. 
If there is no perfect matching in $H$, we report failure.

The two phases of the algorithm are as follows: 

\myparagraph{Phase 1: Assigning to degree-2 vertices}
For a vertex $x \in X$ with degree 2, we denote $\neighbor{G}{x} = \braced{r, c}$, we assign $c \to x$. 
This assignment is safe because the congestion of $(r,x)$ is equal to $\degree{G}{c}$, which is at most $K$ by assumption. 
Moreover, this $x$ cannot be assigned to any other vertices in $C$ which implies that $\cng{G,T}{r,x}$ will remain unchanged. 

\myparagraph{Phase 2: Assigning to degree-3 vertices}
After the first phase, the remaining vertices in $X$ that are available for assignments have degree 3. 
Let $X'$ be the set of such vertices, and $C'$ be the set of unassigned vertices in $C$. 
Unlike in the $\problemKSTC{5}{2}$ algorithm, we cannot directly use a matching from $C'$ to $X'$ to create feasible assignments 
because it is possible for two vertices in $C'$ to be assigned to the same vertex in $X$ (not allowed in the second phase of $\problemKSTC{5}{2}$ algorithm). 
However, we can still capture assigning a pair of vertices in $C'$ to the same vertex in $X'$ by matching this pair to themselves. 
To accomplish this, we reduce the assignment problem from $C'$ to $X'$ to finding a perfect matching in an auxiliary graph $H$ (not necessarily bipartite).
 
The vertices of $H$ consists of all vertices in $X' \cup C'$. In addition, if $|X' \cup C'|$ is odd, we also add $r$ to $H$.  
For each $c \in C'$,  we add an edge $(c,x)$ to $H$ where $x \in \neighbor{G}{c}$ if $\degree{G}{c} + 1 \le K$.  
This condition ensures that $c \to x$ is a feasible assignment.  
Furthermore, for each pair of vertices $c_1, c_2 \in C'$ that share the same neighbor $x \in X'$, if $\degree{G}{c_1} + \degree{G}{c_2} -1 \le K$, we add the edge $(c_1, c_2)$ to $H$. 
This condition is equivalent to $\cng{G,T}{r,x} \le K$ after both assignments $c_1 \to x$ and $c_2 \to x$ have been made. 
Finally, we add edges between any pair of vertices in $X'$ and, if $r$ is in $H$, we add edges from $r$ to all vertices in $X'$. 
Figure~\ref{fig: stc2 matching}a shows an example construction of $H$.  

\begin{figure}[ht]
	\begin{center}
	\includegraphics[height = 1.4in]{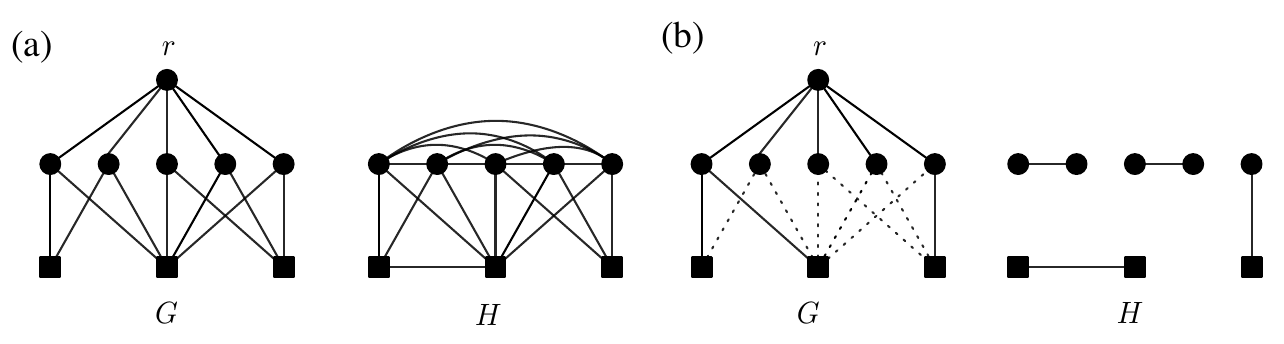}
	\end{center}
	\caption{(a) An example of $H$ constructed from $G$ for algorithm $\problemKDSTC{6}{2}$ in Phase~2. 
		(b) Assignments in $G$ built from a perfect matching in $H$.}
	\label{fig: stc2 matching}
\end{figure}

We proceed to find a maximum matching $M$ in $H$, which can be done in polynomial time~\cite{galil_1986_algorithms_maximum_matching}. 
If $M$ is a perfect matching, it will define assignments for all vertices in $C'$ such that these edges combined with the tree $T$ result in a 
spanning tree of congestion at most $K$ for the graph $G$. 
If $M$ is not a perfect matching, we report failure. 
The following lemma establishes the correctness of this phase:

\begin{lemma}
\label{lemma: STC2 matching}
There exists feasible assignments for all vertices in $C'$ if and only if $H$ has a perfect matching. 
\end{lemma}

\begin{proof}
$(\Rightarrow)$ 
Let $A$ denotes the assignments for vertices in $C'$ that represents a depth-2 spanning tree rooted at $r$ with congestion at most $K$. 
We will show that $H$ admits a perfect matching $M$: 
For each assignment $c \to x$, if $x$ is not assigned to any other vertex in $C'$, we add $(c, x)$ to $M$; 
otherwise, $x$ is assigned to exactly one other vertex $c' \in C'$, we add $(c, c')$ to $M$. 
The remaining vertices that have not been matched are in $X'$ and $r$ (if it is in $H$). 
These vertices can be matched arbitrarily since they form a clique of even size. 

\smallskip
$(\Leftarrow)$ 
Suppose $M$ is a perfect matching of $H$. We make assignments for a vertex $c \in C'$ as follows (refer to Figure~\ref{fig: stc2 matching}b for an example): 
\begin{itemize}
	\item If $c$ is matched with a vertex $x \in X'$ in $M$, we assign $c \to x$. 
	This assignment is feasible by the construction of $H$, and we also know that $x$ cannot be assigned to any other vertex according to the condition of the matching. 
	\item  If $c$ is matched with another vertex $c' \in C'$, then there exists a vertex $x \in X'$ such that $\neighbor{G}{x} = \braced{r, c, c'}$. 
	In this case, we assign both $c, c'$ to $x$. By the construction of $H$, both assignments are feasible, and $x$ is also not used for assignment to any other vertex in $C'$.   
\end{itemize}
This assignment represents a depth-2 spanning tree rooted at $r$ with congestion at most $K$.
\end{proof}


\subsection{$\problemKDSTC{K}{2}$ for bipartite graphs with all degrees in $C$ equal}
\label{subsec: STC2 flow}

We now describe a polynomial time algorithm for $\problemKDSTC{K}{2}$ restricted to
bipartite graph $G$ when all vertices in $C$ have degree $\alpha$, for some positive integer $\alpha$. 
This will prove Theorem~\ref{thm: STC2 algorithm for regular C}. 
We can assume that $\alpha \le K$, for otherwise the congestion of the leaf edges will exceed $K$.
As before, we focus on finding feasible assignments that map each vertex in $C$ to its neighbor in $X$. 
These assignment represent a spanning tree rooted at $r$ whose congestion of all edges in $E_r$ not exceeding $K$.

We first consider the case when $\alpha = 2$. 
For each $c \in C$, if $\neighbor{G}{c} = \braced{x_1, x_2}$, we can assign $c$ arbitrarily to either $x_1$ or $x_2$, because the congestions of 
both $(r,x_1)$ and $(r,x_2)$ are not affected by either assignment. 

From now on, we assume that $\alpha \ge 3$. 
The idea of the algorithm is to express the assignments for vertices in $C$ via the maximum $s-t$ flow in an auxiliary flow network $F$ that can be
constructed in polynomial time from $G$. 
The graph $F$ includes all edges and vertices of $G$. All the edges are directed from $r$ to $X$ and from $X$ to $C$. 
Additionally, $F$ has a source vertex $s$ and directed edge $(s,r)$, and a sink vertex $t$ with directed edges from vertices $C$ to $t$. 
We use  $\capacity{u}{v}$ to denote the capacity of the edge $(u,v)$. 
The capacities of all edges in $F$ are defined as follows: 

\begin{itemize}
	\item $\capacity{s}{r} = |C|$
	\item For each vertex $x \in X$, $\capacity{r}{x}  =  \floor{\dfrac{K- \degree{G}{x}}{\alpha-2}}$
	\item For each edge $(x,c)$ in $F$ where $x \in X$ and $c \in C$, $\capacity{x}{c} = 1$
	\item For each vertex $c \in C$, $\capacity{c}{t} = 1$
\end{itemize}

We then find, in polynomial time, a maximum $s-t$ flow $f$ in $F$. As we will show,
if $f$ has an $s-t$ flow of value $|C|$, this flow will define a feasible assignments for all vertices in $C$ representing a depth-2 spanning tree rooted at $r$ with congestion at most $K$. 
If the maximum flow value is less than $|C|$, we report failure. 
The following lemma establish the correctness of the reduction: 

\begin{lemma}
\label{lemma: stc2 flow}
There exists feasible assignments for all vertices in $C$ if and only if $F$ has an $s-t$ flow of value $|C|$. 
\end{lemma}

\begin{figure}[ht]
	\begin{center}
	\includegraphics[height = 2.2in]{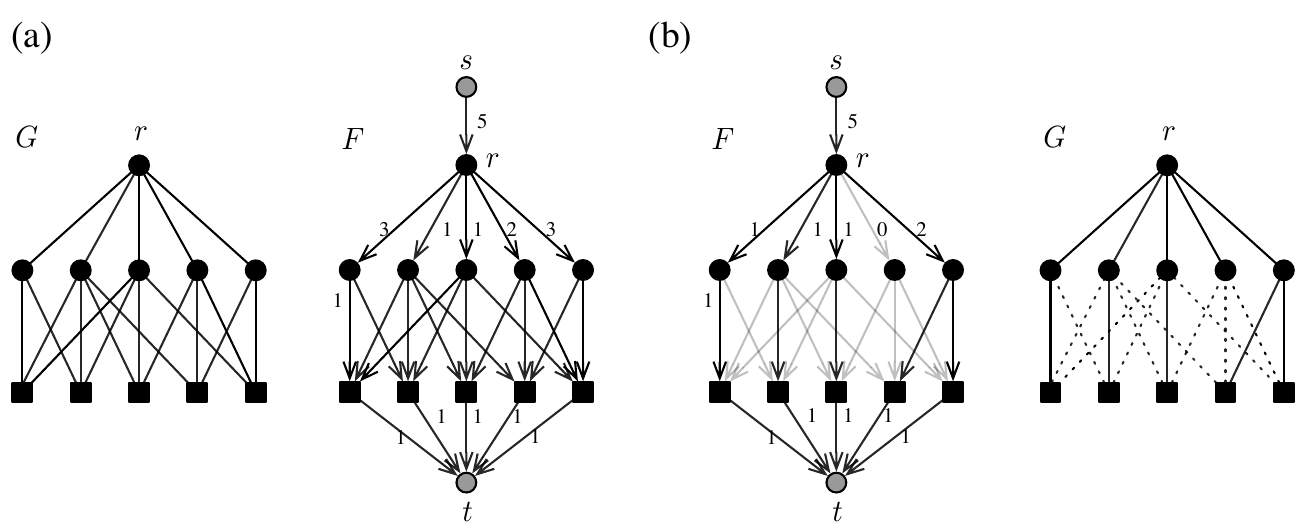}
	\end{center}
	\caption{(a) 
	An example of the auxiliary network $F$ (on the right) constructed from $G$ (on the left). 
	Edges from $X$ to $C$ have capacity $1$, all other edges have capacities as shown. 
	(b) 
	On the left, a maximum flow in $F$. Dark edges have flows with shown values and light edges have no flow.
	On the right, the assignment obtained from this flow.}
	\label{fig: example graph reduction flow}
\end{figure}

\begin{proof}
$(\Rightarrow)$ 
Suppose $F$ has an $s-t$ flow $f$ of value $|C|$. We denote by $\flow{u}{v}$ the flow value on the edge $(u,v)$. 
Since $|C|$ is integral and all capacities are integral, we can assume that flow values of $f$ on all edges are integral.  
Therefore, for each vertex $c \in C$, $\flow{c}{t} = 1$, which implies that there is exactly one vertex $x \in X$ with $\flow{c}{x} = 1$. 
We then assign $c \to x$. 

Next, we need to verify that in the corresponding tree the congestions of the edges in $E_r$ are at most $K$. 
For each vertex $x \in X$, the number of vertices in $C$ that can be assigned to this $x$ is bounded by $\capacity{r}{x}$. 
By Equation~\ref{eqn: cng partial tree}, $\cng{}{r,x} \le \degree{G}{x} + \capacity{r}{x} \, (\alpha - 2) \le K$, which completes the proof of this implication. 

\smallskip
$(\Leftarrow)$ 
Suppose there exist feasible assignments for all vertices in $C$. From this assignment we will construct an $s-t$ flow $f$ for $F$ with value $|C|$.
For each vertex $c \in C$, if $c$ is assigned to $x \in X$, then we define $\flow{x}{c} = 1$ and $\flow{x'}{c} = 0$ for all $x' \in \neighbor{G}{c} \setminus \braced{x}$. 
Next, for each vertex $x \in X$, we define $\flow{r}{x} = n_x$ where $n_x$ is the number of vertices in $C$ that are assigned to $x$. 
Due to the congestion bound on $(r,x)$, we have $n_x \le \frac{K- \degree{G}{x}}{\alpha-2}$. However, since $n_x$ is integral, we have $n_x \le \floor{\frac{K- \degree{G}{x}}{\alpha-2}} = \capacity{r}{x}$. 
Lastly, we let $\flow{c}{t} = 1$ for each $c \in C$ and $\flow{r}{s} = |C|$. 
Clearly, the constructed flow $f$ has value $|C|$, and it satisfies flow conservation and capacity constraints of $F$. 
\end{proof}



%
%
%


\bibliographystyle{plainurl}
\bibliography{stc_references}


\end{document}